\documentclass[11pt, letter]{article}

\usepackage{tikz}
\usepackage{amsmath}
\usepackage{graphicx}
\usetikzlibrary{decorations.pathmorphing}
\usepackage{todonotes}

\usepackage[margin=1in]{geometry}
\usepackage{amsmath,amssymb,amsthm}
\usepackage{hyperref}
\usepackage{cleveref}
\usepackage{cite}
\usepackage{algorithm2e}
\usepackage{enumitem}

\newtheorem{lemma}{Lemma}
\newtheorem{theorem}{Theorem}
\newtheorem{proposition}{Proposition}
\newtheorem{corollary}{Corollary}
\newtheorem{claim*}{Claim}
\theoremstyle{definition}
\newtheorem{definition}{Definition}
\newtheorem{observation}{Observation}

\def\R{\ensuremath{\mathbb{R}}}
\def\Rnonneg{\ensuremath{\mathbb{R}_{\ge 0}}}

\def\E{\ensuremath{\mathbb{E}}}

\def\Opt#1{\ensuremath{\textsc{Opt}(#1)}}
\def\Mech#1#2{\ensuremath{\textsc{Mech}_{#1}(#2)}}

\title{Universally truthful mechanisms for scheduling}
\author{
	Georgios Anastasiadis\thanks{School of Informatics, Aristotle University of Thessaloniki, Greece. Email: \texttt{ganastag@csd.auth.gr}}\and
	George Christodoulou\thanks{School of Informatics, Aristotle University of Thessaloniki, Greece; Archimedes, Athena Research Center, Greece. Email: \texttt{gichristo@csd.auth.gr}}\and 
	Elias Koutsoupias\thanks{Department of Computer Science, University of Oxford, UK. Email: \texttt{elias.koutsoupias@cs.ox.ac.uk}}\and 
	Annamária Kovács\thanks{Institute of Computer Science, Goethe University Frankfurt, Germany. Email: \texttt{panni@cs.uni-frankfurt.de}}\and 
	Conrad Schecker\thanks{Institute of Computer Science, Goethe University Frankfurt, Germany. Email: \texttt{schecker@em.uni-frankfurt.de}}
}
\date{}

\begin{document}
	\maketitle

	\begin{abstract}
		We consider universally truthful randomized mechanisms for the problem of
scheduling $m$ jobs on $n$ unrelated machines. We prove a lower bound on the
expected approximation ratio of every such mechanism whose probability distribution has
discrete support. We show that no universally truthful randomized
mechanism in this class can achieve approximation ratio smaller than 
$n/12 - o(n)$ with respect to the optimal makespan. We match this, up to a constant factor, 
by a mechanism with approximation ratio $n/2 + o(n)$.

	\end{abstract}
    
	\section{Introduction}

We study truthful randomized mechanisms for the classic problem of scheduling on
unrelated machines. In this problem, there are $m$ tasks to be assigned to $n$
machines, and each machine $i$ has a processing time $t_{ij}$ for
executing task $j$. The objective is to minimize the \emph{makespan}, defined as the maximum completion time among all machines.

Unlike traditional algorithms, a truthful mechanism must incentivize
each machine $i$ by appropriate payments to reveal its \emph{private
  values} $t_{ij}$. This additional restriction places the problem in
the general research area of \emph{mechanism design}, a branch of game
theory and microeconomics. In this area, the central objective is to
design algorithms (called mechanisms) in strategic environments, where
the input is held privately by selfish participants. The computational
and information-theoretic aspects of mechanism design were first
studied by Nisan and Ronen who, in their seminal
work~\cite{NisanRonen2001}, laid the foundations of the field of
\emph{algorithmic mechanism design}. They proposed the scheduling
problem as the canonical algorithmic problem in this area and posed
the question of the best approximation ratio that can be achieved by
truthful deterministic and randomized mechanisms. Two aspects make
this a challenging question: First, the processing times $t_{ij}$ of
each machine, which form the input of the mechanism, are
\emph{multidimensional}, which places the problem in the complex and
poorly-understood domain of multidimensional mechanism design. Second,
the objective is of \emph{min-max} type, which is substantially more
demanding than the well-studied min-sum objective which is truthfully
optimized by the VCG mechanism.

The status of this problem for \emph{deterministic} mechanisms was
recently settled by Christodoulou, Koutsoupias and Kov\'acs~\cite{ChristodoulouKK26}: the best approximation ratio for the
scheduling problem achieved by truthful mechanisms is exactly $n$,
equal to the number of machines. The upper bound comes from the
observation that the VCG mechanism, which truthfully minimizes the sum
of completion times, rather than the makespan, has approximation ratio
$n$. On the other hand, the lower bound is significantly more demanding and was
based on a combination of a sophisticated class of instances, a partial
characterization of mechanisms for two machines and two tasks, and an
involved induction on the number of machines~\cite{ChristodoulouKK26}.

In light of this negative result, shifting to \emph{randomized
  mechanisms} is the natural next step. Randomization in mechanism
design, which has been used since the early work of Nisan and
Ronen~\cite{NisanRonen2001}, comes in two flavors: universally
truthful and truthful-in-expectation mechanisms. A randomized
mechanism is called \textit{universally truthful} if it is defined as
a probability distribution over deterministic truthful mechanisms. A
randomized mechanism is \emph{truthful-in-expectation} when agents
cannot strictly improve their expected utility by misreporting. The class of
truthful-in-expectation mechanisms is broader, but the universally
truthful mechanisms have the important characteristic that they are
truthful even ex post (that is, even when the random coins become
known). A closely related model is the model of \emph{fractional}
mechanisms, where each task can be fractionally allocated to multiple
machines.

Unlike the case of deterministic mechanisms, the situation for
randomized mechanisms is highly unresolved across all these models. 
The best known mechanism for each of
the above models achieves an $O(n)$ approximation
guarantee~\cite{MualemSchapira2018,LuYu2008Randomized,LuYu2008Improved}. This
is far from the best lower bound which is $2-1/n$~
\cite{MualemSchapira2018,ChristodoulouKK26}.

In this work, we resolve this problem for universally truthful
mechanisms, when they are defined as a discrete probability distribution
over deterministic mechanisms. Our main result establishes a linear
lower bound for all universally truthful mechanisms.

\begin{theorem} \label{thm:nrtheorem} There is no universally truthful
  randomized mechanism with approximation ratio better than $\Omega(n)$ for
  the problem of scheduling $n$ unrelated machines.
\end{theorem}

In particular, we show that the approximation ratio is at least $(3/2-\sqrt{2})(n-1) \approx 0.085 (n-1) > (n-1)/12$.

On the \emph{positive side}, we provide a universally truthful
mechanism that achieves an approximation ratio of $n/2 + o(n)$.  This is an asymptotic improvement over the linear coefficient of the
previous best bound of $0.837n$ by \cite{LuYu2008Improved}. Our
bound asymptotically matches the $(n+1)/2$ bound for 
fractional~\cite{ChristodoulouKoutsoupiasKovacs2010Fractional} and the $(n+5)/2$ bound for truthful-in-expectation mechanisms~\cite{LuYu2008Randomized}. 
It also asymptotically matches the lower
bound of $(n+1)/2$ established for \emph{task-independent}
mechanisms~\cite{ChristodoulouKoutsoupiasKovacs2010Fractional,LuYu2008Randomized}.

We focus on universally truthful mechanisms defined as \emph{discrete probability distributions} over deterministic mechanisms. We make this restriction primarily to avoid the measure-theoretic technicalities associated with continuous distributions. Because both discrete and continuous distributions share the same underlying truthfulness properties in the scheduling domain, we expect that our approach and the resulting linear lower bound naturally extend to distributions with uncountable support.

The remainder of the paper is organized as follows. Section 2 formally defines the scheduling problem, the mechanism design framework, and the fundamental concept of weak monotonicity. Section 3 contains the proof of our main negative result; we introduce the graph-theoretic representation of the instances, detail our distribution over multi-clique instances, and apply Yao's minimax principle to establish the $\Omega(n)$ lower bound. In Section 4, we present our positive result by introducing the \textsc{Exp-Bounded-Square} mechanism and proving both its universal truthfulness and its expected approximation guarantee. Finally, the formal adaptation of the Box Theorem, which forms the geometric core of our lower bound analysis, is given in the Appendix.

\subsection{Related Work}

\paragraph{Deterministic Mechanisms.}
The seminal work of Nisan and Ronen~\cite{NisanRonen2001} established a lower bound of $2$ for deterministic truthful mechanisms and showed that the VCG mechanism yields an $n$-approximation, leaving a huge gap. A sequence of works improved the lower bound only incrementally: first to $2.41$~\cite{ChristodoulouKoutsoupiasVidali2009}, then to $2.61$~\cite{KoutsoupiasVidali2013}, later to $2.75$~\cite{GiannakopoulosHammerlPocas2020}, and eventually to $3$~\cite{DobzinskiShaulker2020}. Christodoulou, Koutsoupias, and Kov\'acs~\cite{ChristodoulouKK26} fully closed this gap by showing that no deterministic truthful mechanism can circumvent the $n$-approximation barrier, thereby resolving the long-standing Nisan--Ronen conjecture.

\paragraph{Randomized Mechanisms.}
Mu'alem and Schapira~\cite{MualemSchapira2018} showed a lower bound of $2-1/n$ for both truthful-in-expectation and universally truthful randomized mechanisms. Christodoulou, Koutsoupias, and Kov\'acs~\cite{ChristodoulouKoutsoupiasKovacs2010Fractional} extended this lower bound to fractional mechanisms, where each task can be fractionally allocated to multiple machines. They also presented a (deterministic) fractional mechanism with a guarantee of $(n+1)/2$. Building on this framework, Lu and Yu~\cite{LuYu2008Randomized} obtained a truthful-in-expectation mechanism achieving an approximation guarantee of $(n+5)/2$.

Regarding universally truthful mechanisms, Nisan and Ronen~\cite{NisanRonen2001} proposed a mechanism for the case of two machines achieving a $7/4$-approximation. Mu'alem and Schapira~\cite{MualemSchapira2018} extended this mechanism to $n$ machines with an approximation guarantee of $0.875n$, which Lu and Yu~\cite{LuYu2008Improved} subsequently improved to $0.837n$. The case of two machines is largely resolved~\cite{LuYu2008Randomized,ChenDuZuluaga2015}.

\paragraph{Bayesian Setting.}
In the Bayesian model, Daskalakis and Weinberg~\cite{DaskalakisWeinberg2015} gave a polynomial-time Bayesian incentive-compatible mechanism whose expected makespan is within a factor of $2$ of the optimal Bayesian truthful mechanism (this guarantee is with respect to the optimal truthful mechanism, rather than the ex post optimal makespan). Chawla, Hartline, Malec, and Sivan~\cite{ChawlaHartlineMalecSivan2013} studied prior-independent mechanisms, where processing times are drawn from an unknown distribution, and provided approximation guarantees under distributional assumptions. Giannakopoulos and Kyropoulou~\cite{GiannakopoulosKyropoulou2017} analyzed the performance of the VCG mechanism in a Bayesian scheduling model under independence and symmetry assumptions on the processing-time distributions, showing that VCG achieves an $O(\log n / \log \log n)$ approximation.

\paragraph{Special Cases.}
Several variants and special cases of truthful scheduling have also been studied, yielding significant results. Ashlagi, Dobzinski, and Lavi~\cite{AshlagiDobzinskiLavi2012} showed a lower bound of $n$ for \textit{anonymous mechanisms}, a natural class of mechanisms that treat machines symmetrically. Lavi and Swamy~\cite{LaviSwamy2009} studied a restricted, yet multidimensional, input domain in which each processing time can take only one of two possible values, `low'' or `high''. For this domain, they gave a deterministic truthful mechanism with an approximation factor of $2$. They further showed that, even in this restricted setting, no truthful mechanism can achieve the optimal makespan, proving a lower bound of $11/10$. Yu~\cite{Yu2009} extended this line of work to domains with a range of possible values. Auletta, Christodoulou, and Penna~\cite{AulettaChristodoulouPenna2015} studied multidimensional domains in which each machine's private information is encoded by a single bit. Christodoulou, Koutsoupias, and Kov\'acs~\cite{ChristodoulouKK25} investigated \textit{graph-restricted} instances, in which each task can be assigned to at most two machines. Their graph-based perspective is closely related to the present lower-bound construction, which relies on multi-cliques and multi-stars.

\paragraph{Learning-Augmented Mechanisms.}
Finally, a recent line of research investigates learning-augmented mechanisms under various prediction regimes, ranging from full-input predictions~\cite{XuL22, BalkanskiGT23} to minimal output-only predictions~\cite{CSV24}. Within the full-prediction paradigm, Balkanski, Gkatzelis, and Tan~\cite{BalkanskiGT23} provided a deterministic mechanism achieving a constant-factor approximation under accurate predictions and an $O(n)$ robustness bound under arbitrary errors. Cole, Gupta, and Jangir~\cite{ColeGJ25} achieved similar guarantees using a more limited intermediate model that requires predicting only $O(m+n)$ structural values, whereas the robustness bound of Christodoulou, Sgouritsa, and Vlachos~\cite{CSV24} for output-only predictions degrades to $O(n^2)$.

	\section{Preliminaries}
\label{sec:prelim}

We consider the scheduling problem for unrelated machines.
There is a set $M$ of $m$ \emph{tasks} that need to be scheduled on a set $N$ of $n$ \emph{machines}.
Each task $j \in M$ has a \emph{processing time} (or \emph{cost})  $t_{ij} \in \Rnonneg$ for every machine $i \in N$.
If $T \subseteq M$ is the subset of tasks that are scheduled on machine $i,$ then the \emph{completion time} of this machine is  $t_i(T) := \sum_{j \in T} t_{ij}$.
The goal is to find an allocation of all tasks to the machines such that the maximum completion time over the machines, called \emph{makespan}, is minimized.

\subsection{Mechanism Design Setting}

Each individual machine $i\in N$ is controlled by a selfish \emph{agent (player)}, and the vector of running times $t_i=(t_{i1},t_{i2},\ldots, t_{im}),$ also called the \emph{type} of agent $i,$ is privately known to that agent. We will use the terms \emph{machine, agent} and \emph{player} interchangeably. The set $\mathcal T_i$ contains all possible types of $i$, i.e., all possible vectors $t_i\in \Rnonneg^m.$ Additionally, let $\mathcal T = \times_{i\in N}\mathcal T_i$ denote the space of possible type profiles for all the $n$ agents. 
In general, a mechanism defines for each agent $i$ a set $\mathcal B_i$ of all available strategies that the agent can choose from. We consider \emph{direct-revelation} mechanisms and thus the strategy space of agent $i$
coincides with its type space, that is, $\mathcal B_i=\mathcal T_i$. An agent may
report a bid vector $b_i=(b_{i1}, \dots, b_{im})\in\mathcal T_i$ that differs from its true type $t_i$, if this serves their interest.

\subsection{Deterministic Mechanisms}
A deterministic mechanism $ M = (A,\mathcal P)$ consists of two parts: 
\begin{itemize}
	\item An \textbf{allocation algorithm} $A\,$:  Let $\mathcal A$ denote the set of all
	partitions of the tasks among the machines. The \emph{allocation algorithm (allocation rule)} $A$ 
	is a function
	$A:\mathcal T\to \mathcal A$.
	Given a bid vector $b\in\mathcal T$, the allocation $A(b)$ assigns each task
	to exactly one machine. We denote by $A_i(b)$ the set of tasks assigned to machine $i$ under bids $b$.
	
	\item  A \textbf{payment scheme} $\mathcal P\,$: The \emph{payment scheme} $\mathcal{P}= (\mathcal P_1,\dots,\mathcal P_n)$, specifies, for every bid vector $b$, the payments to the players. In particular, the payment of each agent $i$ is determined by the   function $\mathcal P_i: \mathcal T \rightarrow \mathbb{R}$.

\end{itemize}
The \emph{utility} $u_i^M$ of a player $i$ under the deterministic mechanism $M$ is the payment that he gets from the mechanism, minus the \emph{actual} time they need to process the set of tasks assigned to him, i.e., $u_i^M(b)= \mathcal P_i(b)- t_i(A_i(b))$ where $t_i(A_i(b))=\sum_{j\in A_i(b)}t_{ij}$. The mechanism $M$ is \emph{truthful} if for every player, reporting their true type is a dominant strategy. Formally, for every agent $i\in N$, every pair of types
	$t_i,b_i\in\mathcal T_i$, where $t_i$ is the true type of agent $i$, and every $b_{-i}\in\mathcal T_{-i}$ it holds that: $$u_i^M(t_i,b_{-i})\ge u_i^M(b_i,b_{-i}).$$

	 The performance of $M$ is measured by the makespan $\Mech{A}{t}:=\max_{i\in N}t_i(A_i(t))$ achieved by its allocation algorithm $A$ for input type $t$, which is related to the optimal makespan $\Opt{t}:=\min_{A\in \mathcal A}\max_{i\in N}t_i(A_i)$. 
	 We define the approximation ratio of $A$ as:
	 
	\[
	\rho(A,t) := 
	\begin{cases}
		\frac {\Mech{A}{t}} {\Opt{t}} , & \text{if}~\Opt{t} > 0,\\
		1, & \text{if}~\Opt{t}=0~ \text{and} ~\Mech{A}{t}=0,\\
		\infty, & \text{if}~\Opt{t}=0 ~\text{and}~ \Mech{A}{t}>0.
	\end{cases}
	\] 
	
	We call a mechanism \emph{$\beta$-approximate}, with
        $\beta\in \mathbb R_{+}$, if $\rho(A,t) \leq \beta $ for every
        $t\in \mathcal T$. For an allocation algorithm $A$ with
        \emph{unbounded} approximation ratio it holds that for every
        $\beta\in \mathbb R_{+}$, there is $t\in \mathcal T$ such that
        $\rho(A,t)> \beta$. We will mainly consider mechanisms with
        bounded approximation ratio.

		\subsection{Weak Monotonicity}
                A deterministic mechanism consists of two components:
                an allocation algorithm and a payment scheme.
                Our analysis focuses on the approximation ratio of the
                allocation algorithm, so we will make use of the
                following characterization of allocation algorithms
                which can be implemented truthfully (by some appropriate payment
                scheme).

	 \begin{definition}
	 	An allocation algorithm $A$ is called \emph{weakly monotone (WMON)} if it satisfies the following property: for every $i\in N$ and every two input types $t=(t_i,t_{-i})$ and $t'=(t'_i,t_{-i})$, the corresponding allocations $A_i$ and $A_i'$ satisfy $$t_i(A_i)-t_i(A'_i)\le t'_i(A_i)-t'_i(A'_i). $$
	 	An equivalent condition, using the indicator variables $a_{ij}(t)\in \{0,1\}$ of whether task $j$ is allocated to player $i$, is $$ \sum_{j\in M} (a_{ij}(t')-a_{ij}(t))(t'_{ij}-t_{ij})\le 0. $$

	 	\end{definition}
	 	
	 	It is well known that the allocation function of every truthful mechanism is weakly monotone~\cite{BikhchandaniChatterjiLaviMualemNisanSen2006}. Also, it is a sufficient condition for truthfulness in every convex domain~\cite{SaksYu2005}, and therefore in the scheduling domain.
	 	
	 	We will now consider some immediate consequences and implications of weak monotonicity that will be useful tools in our proof. The first is that when we fix the values of the players for a subset of tasks, then  the restriction of the allocation to the rest of the tasks is still weakly monotone. Formally:
	 	\begin{lemma}
	 		\label{lem:restricted-subset-WMON}
	 			Let $A$ be a weakly monotone allocation, and let $(T,M\setminus T)$  be a partition of $M$. When we fix the values of the players for the tasks in $M\setminus T$, the restriction of the allocation $A$ to $T$ is also weakly monotone.
	 	\end{lemma}

	 	See \cite{ChristodoulouKoutsoupiasKovacs2020Submodular} for a simple proof of Lemma~\ref{lem:restricted-subset-WMON}. The following lemma, first used in~\cite{NisanRonen2001},
	 	is a standard tool for showing lower bounds for truthful mechanisms (see~\cite{AshlagiDobzinskiLavi2012,
		ChristodoulouKoutsoupiasVidali2009,DobzinskiShaulker2020,
	GiannakopoulosHammerlPocas2020,MualemSchapira2018,NisanRonen2001} and refer to~\cite{NisanRonen2001} for a proof).
	\begin{lemma}(\cite{NisanRonen2001}) 
		Consider a truthful mechanism $(A,\mathcal{P})$ and its allocation
		for a bid vector $t$. Let $T$ be a subset of the tasks allocated to player $i$
		and let $T'$ be a subset of the tasks allocated to the other players. Consider
		any bid profile $t' = (t'_i,t_{-i})$ that is obtained from $t$ by decreasing the
		values of player $i$ in $T$, i.e., $t'_{ij} < t_{ij}$, $j \in T$, and increasing
		the values of player $i$ in $T'$, i.e., $t'_{ij} > t_{ij}$, $j \in T'$. Then the
		allocation of player $i$ for $t$ and $t'$ agree for all tasks in $T \cup T'$.
	\end{lemma}

	Notice that only the set $T$ of tasks allocated to player $i$
        remains the same. This does not preclude changing the
        allocation of the other players for the tasks in
        $M\setminus T$, unless there are only two
        players. \emph{Graph settings} with only two possible
        players for each task (as defined in Section~\ref{sec:graphs}, see also
        \cite{ChristodoulouKK25,ChristodoulouKoutsoupiasKovacs2020Submodular,
          ChristodoulouKK26}) are helpful to overcome this hurdle.
	
	\subsection{Randomized Mechanisms}
        In this paper, we restrict attention to \emph{universally truthful
   randomized mechanisms}, that is, randomized mechanisms whose support consists
   only of truthful deterministic mechanisms. Since the performance (i.e., the expected makespan) of such mechanisms depends only on
   the allocation rules of the mechanisms in their support, we may equivalently view such
   randomized mechanisms as distributions over deterministic allocation rules that are
   implementable by truthful mechanisms. By \cite{BikhchandaniChatterjiLaviMualemNisanSen2006} every such allocation rule satisfies the weak monotonicity condition. Thus, we can view universally-truthful randomized mechanisms as distributions over deterministic weakly monotone allocation algorithms.

   Let $\mathcal A$ denote the set of all deterministic weakly monotone allocation
   algorithms, and let $\mathcal R$ denote the set of all \emph{discrete} probability
   distributions over $\mathcal A$. For every  reported bid profile  $b\in \mathcal T$, a randomized mechanism $R\in\mathcal R$ samples
   an allocation rule $A \sim  R$ and applies it to $b.$ 
   Let $$\operatorname{supp}(R):= \{A'\in \mathcal  A\;|\; \mathbb P_{A \sim R}(A=A')>0\}$$ denote the \emph{support of $R$}.

   For every $R\in \mathcal R$ and $t \in \mathcal T$, let
   $\Mech{R}{t} := \E_{A \sim R}[\Mech{A}{t}]$ denote the expected
   makespan of $R$ for input $t$.  The \emph{(expected) approximation
     ratio} of $R$ for a given type $t$ is defined
   as $$\rho(R,t)=\mathbb{E}_{A\sim R}[\rho(A,t)].$$ Note that, if
   $\Opt{t}>0,$ then $\rho(R,t)=\Mech{R}{t}/\Opt{t}.$ We call a
   \emph{randomized mechanism $R$ $\beta$-approximate} if
   $\rho(R,t)\leq \beta$ for every $t\in \mathcal T.$

	\section{Lower Bound}
\label{sec:lb}

In this section, we derive our main result by describing a
distribution over instances and showing a lower bound on the ratio of
the expected makespan of any given deterministic weakly monotone
algorithm $A$ divided by the expected optimum over this input
distribution. 

In Section~\ref{sec: Yao} we show that it suffices to
focus on universally truthful mechanisms with finite support. This allows us to apply the Yao Principle assuming a 'finite game' between algorithm designer and adversary. Moreover, having a finite support is necessary so that there trivially exists a \emph{common} upper bound $K$ on the approximation ratio of the considered deterministic algorithms (in the support of the randomized algorithm). The parameter $K$ affects the size of input instances used in the instance distribution.
In Section~\ref{sec:graphs} we describe \emph{graph instances} for the scheduling problem, and  the \emph{Box Theorem} from~\cite{ChristodoulouKoutsoupiasKovacs2023Proof} used in the lower bound proof. 
 Then,
in Section~\ref{sec:random-instances} we define the distribution of
the random instances. 
In the last two sections, we analyze the
approximation guarantee, first based on an analysis for an arbitrary fixed \emph{star} in Section
\ref{sec:star} and finally for a \emph{random star} in Section
\ref{sec:random-star}.

\subsection{Applying Yao's Principle}
\label{sec: Yao}

The main result of this work is a lower bound on the approximation
ratio of universally truthful mechanisms. We show a lower bound
of $\Omega(n),$ using Yao's Minimax Principle\cite{Yao77, MotwaniR95}.

Obviously, the approximation ratio of a randomized mechanism $R$ is
unbounded, if $\operatorname{supp}(R)$ contains a (single) deterministic
mechanism with unbounded ratio. To see this, assume that such a
mechanism $A$ occurred in $R$ with probability $p>0.$ Then for an
arbitrarily large $\beta\geq 1$ there would exist an instance $t_\beta$
with approximation ratio $\rho(A,t_\beta)\geq \beta/p,$ and
$\rho(R, t_\beta)\geq p\cdot (\beta/p)=\beta$ would hold.  Therefore,
it is enough to show a lower bound only for such $R$ that are (discrete)
distributions over deterministic mechanisms $A$ with bounded
ratio.

However, in order to apply Yao's Principle directly, we need to
assume the stronger condition that a \emph{common} bound $K$ holds for
every algorithm $A$ in $\operatorname{supp}(R).$ For this  reason, \emph{and}  in order to have  a ``finite game'' for Yao's argument, we assume \emph{finite}
support  of $R.$ We define for each $K$ a corresponding set of randomized allocations $\mathcal R_K.$ Then, we show in Proposition~\ref{prop:discrete}, that it is enough to prove the lower bound for randomized mechanisms from $\mathcal R_K$ (for arbitrary given $K$).

Let $\mathcal A_K$ denote the set of deterministic allocation rules
with approximation factor \emph{less than} $K;$ we define the set of
randomized mechanisms $\mathcal R_K$ as
$$\mathcal R_K := \left\{R \in \mathcal R\quad \vert \quad
  \operatorname{supp}(R)\subseteq \mathcal A_K,\;\;
  |\operatorname{supp}(R)|<\infty \right\},$$ that is, the
distributions having only finitely many, $K$-approximate, truthful,
deterministic mechanisms in their support. In Section~\ref{sec:random-instances} we show that for
an arbitrary given $K$, there is a distribution
$\mathcal G^*$ over instances $\mathcal I^*$ with the
property that for every
$A\in \mathcal A_K$
$$\frac{\mathbb E_{\mathcal I^*\sim \mathcal G^*}[\Mech{A}{\mathcal
    I^*}]}{\mathbb E_{\mathcal I^*\sim \mathcal G^*}[\Opt{\mathcal
    I^*}]}>\frac{n-1}{12}.$$ We remark that $\mathcal G^*$ will be a different distribution over instances for different $K$ parameters, because the number of tasks in every $\mathcal I^*$ from the support of $\mathcal G^*,$ will
  increase with increasing $K.$ Then, by applying Yao's
Principle~\cite{Yao77,MotwaniR95}, this implies (see also
Theorem~\ref{thm:mainLB}) that for every $R\in \mathcal R_K$ there exists an
instance $\mathcal I_{R}$ such that
$$\rho(R, \mathcal I_{R})\geq \frac{n-1}{12}.$$

The next proposition shows how this lower bound can be
extended to the case of general discrete distributions.

\begin{proposition} Suppose that for every $K>n+1,$ every universally
  truthful mechanism $R\in \mathcal R_K$ has approximation ratio
  larger than $\beta$. Then every universally truthful mechanism $R$
  has approximation ratio of at least $\beta.$
\label{prop:discrete}
\end{proposition}

\begin{proof} Let $R$ be a given universally truthful mechanism and let $\delta>0$ be
 an arbitrary small constant. Consider the (finitely many) deterministic mechanisms in
  $\operatorname{supp}(R),$ having the highest probabilities of being
  selected with total probability that sums up in at least $1-\delta.$ Let
  $R'$ denote the corresponding distribution, obtained from $R$
  restricted to these mechanisms, with the appropriate normalization
  of their probabilities. Clearly $R'$ has finite support.

  Pick an instance $\mathcal I(\beta,K_\delta)$ that shows approximation
  ratio higher than $\beta$ for $R'$, as defined
  for the highest approximation ratio $K_\delta$ among these finitely
  many mechanisms. Then $R$ has (expected)
  approximation factor of at least $ (\beta(1-\delta)+ 1\cdot \delta)$ which goes to $\beta,$ when 
  $\delta$ goes to $0.$
\end{proof}

\subsection{Graph Instances and the Box Theorem}
\label{sec:graphs}

We generally consider instances that are described by
\emph{multi-graphs} with a vertex set $N$ and an edge set $E:=M.$ The edges
correspond to the tasks and the vertices correspond to the machines
(cf.~\cite{ChristodoulouKK25,ChristodoulouKK26}).  An edge $e $ between vertices
$u, v \in N$ models a task that has very high processing
times\footnote{E.g., processing time $K^2$ will suffice. }  on all
machines except for $u$ and $v$, such that every allocation where $e$
is not processed by either $u$ or $v$ has approximation ratio of at
least $K.$  

Let $Q_{\{u,v\}}$ denote the set of
parallel edges between any two different vertices $u,v\in N.$ Each of these edges represents a different
task and comes with two values for the processing times on machines
$u$ and $v$, respectively. In particular, we will use a
\emph{multi-clique} with edge-multiplicity $\ell$ for an appropriate
$\ell;$ then select a random \emph{multi-star}, and finally a random
\emph{star} from the edges of the multi-clique.

An input $t$ is called \emph{trivial}, if for every task $j\in E$ holds that $t_{ij}=0$ for at least one $i\in N.$ Trivial instances are useful, because if we increase the values on a subset of tasks (say, on a star), then the approximation ratio of $A$ restricted to this subset, gives a lower bound on the optimal approximation ratio altogether.

Consider some edge $e\in Q_{\{u,v\}}$ in an arbitrary multigraph instance $\mathcal I.$ If we fix the values of every other edge, then the allocation of edge $e$ by mechanism $A$ is determined by a so called \emph{boundary function} $\psi^{e,\mathcal I}_{u,v}: \mathbb R_{\geq 0} \rightarrow \mathbb R_{\geq 0},$ as follows (we omit the upper indices $^{e,\mathcal I}$): Let $(t_u, t_v)$ denote the processing times of $u$ and $v$ for task $e,$ then $u$ receives task $e$ if $t_u<\psi_{u,v}(t_v),$ and $v$ receives it if $t_u>\psi_{u,v}(t_v).$ 
The boundary functions are non-decreasing, and it is easy to see that,  $\psi_{v,u}\equiv \psi^{-1}_{u,v}$ holds (except for  discontinuity points). 
A straightforward observation about boundary functions for trivial instances is the following:

\begin{lemma} \label{obs:alphasNew} {\cite{ChristodoulouKK26}} Let $\mathcal I$ be a trivial instance. If  the approximation factor of $A$ is less than $K,$ then for every $u,v\in N,$ every task $e\in Q_{\{u,v\}}$ and every running time $t_v=t_v^e$ for vertex $v,$ it holds that 
\begin{enumerate}
\item[(i)] $t_v/K<\psi_{u,v}(t_v)<K \cdot t_v;$%
\item[(ii)] $\lim_{t_v\rightarrow 0}\psi_{u,v}(t_v)=\psi_{u,v}(0)=0.$ 
\end{enumerate}

\end{lemma}

Consider a (trivial) instance, and a boundary function $\psi_{u,v}()$ of some edge $e\in Q_{\{u,v\}}$ with running times denoted by $(t_u,t_v).$ Since some favorable properties of the allocation $A$ do not hold for $t_v$ values where $\psi_{u,v}()$ is not continuous, we would like to avoid such $t_v$ points when defining our instances. Since the $\psi_{u,v}()$ are non-decreasing, fortunately (for fixed $\mathcal I$) each $\psi_{u,v}$ has at most countably many discontinuity points. Also, $t_v=0$ is never a discontinuity point by Lemma~\ref{obs:alphasNew} (ii). It is known that if the nonzero costs of a trivial multi-graph instance are all independently, randomly perturbed, then discontinuity points of all the  boundary functions (w.r.t. the new, perturbed instance(!)) are avoided~\cite{ChristodoulouKK26}. We say that such an instance \emph{fulfills the continuity requirement}:

\begin{definition}[Continuity requirement~\cite{ChristodoulouKK26}]
\label{def:continuity}
  Fix a mechanism and consider a multi-graph instance $\mathcal I.$ We say that $\mathcal I$ \emph{has a discontinuity} if there exists some edge $e\in Q_{\{u,v\}}$ with costs $(t_u,t_v)$ so that $t_v\neq 0$ and $\psi_{u,v}(\cdot)$ is discontinuous at $t_v$ (or analogously for $t_u$). %
  We say that an instance \emph{satisfies the continuity requirement} if \emph{no rational translation of $\mathcal I$ has a discontinuity}.\footnote{In a \emph{rational translation} every running time  may be shifted, each by arbitrary different rational numbers. }
\end{definition}

An important role in our lower bound proof will be played by star (partial) instances $S$ that are so called \emph{boxes} (or $\varepsilon$-\emph{boxes}) for the given allocation rule $A.$ We define the \emph{box} property next, see also Figure~\ref{fig:box}. 

\begin{definition}
Let the vertices/machines $u$ and $v_1, v_2,\ldots, v_k $ be the \emph{root} and \emph{leaf} vertices of a star $S$ in some given multi-clique instance $\mathcal I.$ Let $e_1, e_2, \ldots , e_k$ be the respective edges of the star, and $(t_{v_1}, t_{v_2},\ldots, t_{v_k})$ strictly positive running times of the leaf vertices on them, respectively. The star $S$ is an \emph{$\varepsilon$-box}, when setting  $t_{u}^{e_j}=\psi_{u,v_j}(t_{v_j})-\varepsilon$ for \emph{every} $j=1,2,\ldots, k$ \emph{at once,} the mechanism $A$ allocates (at least) each of the tasks $e_1, e_2, \ldots , e_k$ to the root player.
\end{definition}

Our proof uses the \emph{Box Theorem} (Theorem~\ref{thm:box-restated}), originally proved in \cite{ChristodoulouKK26}. Let $\mathcal I$ be  a trivial  multi-star (or multi-clique) instance, and let $\mathbb P$ denote the probability that a randomly selected star (among the edges having zero costs for the root vertex, and positive costs for the leaf vertices), is an $\varepsilon$-box. The Box-Theorem essentially states that $\mathbb P \rightarrow 1,$ when $\ell \rightarrow \infty,$ where $\ell$ denotes the edge-multiplicity. 
The precise definitions and an adapted version of the Box Theorem are given in the Appendix (Section~\ref{sec:box}).

\subsection{Definition of Random Instances}
\label{sec:random-instances}

The distribution over instances that we use for Yao's argument, can
be viewed as a carefully constructed random multi-graph. We describe
the construction for a given number of machines $n,$ approximation bound
$K$ of the deterministic mechanisms $A\in \mathcal A_K,$ and
parameters $\varepsilon\in (0,\frac{1}{K+1})$ and an even number
$\ell=\ell(\varepsilon)$ as edge-multiplicity, both affecting the
precision of the lower bound (see Lemma~\ref{lem:LBMakespan}).

The basis will be a multi-clique graph $\mathcal{MC}=(N,E)$ over the
players as vertices, and $\ell$ edges for each pair of different
players, that is, altogether $|E|=\ell n(n-1)/2$ tasks. All costs of a
task on other machines than the two endpoints, we set to $K^2$ in every
considered input matrix. Furthermore, we fix a mapping
$\mathcal I_s: E\rightarrow (1-\varepsilon, 1)$ that associates an
independently and uniformly selected value
$s^e \in (1-\varepsilon, 1)$ with every single edge $e\in E.$ These
will be perturbed processing times (as compared to processing time
$1$) for one of the two endvertices of $e.$ This randomization
does \emph{not} play any role in Yao's argument, and we fix
$\mathcal I_s$ \emph{before} choosing the random instance
$\mathcal I^*$ by the distribution $\mathcal G^*$ used for the Yao
proof, (in particular, $\mathcal G^*$ will have finite support). We need
the perturbed $s^e$, only in order for the Box Theorem to hold.

For the analysis later, it is useful to define a preliminary (part of the) distribution $\mathcal G(n, \ell, \varepsilon)$.
 For each vertex pair $u\neq v,$ select a subset of exactly half of their parallel edges $Q_{\{u,v\}}$ uniformly at random, and denote it by $E_{u,v}$.
For every $e \in E_{u,v}$, the processing time for vertex $v$ is set to $s^e(\approx 1),$ and the processing time for vertex $u$ is set to $0.$
Conversely, the remaining half of the parallel edges that were not selected before is denoted by $E_{v,u}$, and the processing times are defined symmetrically: $0$ for vertex $v$, and $s^e$ for vertex $u.$
This describes our preliminary instance (later referred to as \emph{random multi-clique}) $\mathcal I\sim \mathcal G(n, \ell, \varepsilon)$.

Next, we will randomly select a star within the random multi-clique.
To that end pick a vertex $i \in N$ uniformly at random, which we call the \emph{root}.
Every other vertex $j \neq i$ is called \emph{leaf}.
For every leaf $j$, select an edge $e_j \in E_{i,j}$ (i.e., with processing time 0 for $i$), uniformly at random.
By $t_j$ we denote the processing time of $e_j$ for root $i$, and by $s_j$ we denote the processing time of $e_j$ for leaf $j.$
By construction it holds $t_j = 0$ and $s_j = s^{e_j}.$
The set of selected edges $\{e_j\}_{j \neq i}$ form a star subgraph, denoted by $S$.

Finally, for an arbitrary given parameter $c\in (0, 1/2)$ (to be optimized later), we increase the costs of the root player for all the star edges from $t_j=0$ to $t_j^*=c\quad (\forall j :  j\neq i).$ The obtained instance is $\mathcal I^*,$ and the distribution over all these instances (for fixed $\mathcal I_s$)  will be denoted by  ${\mathcal G^*(n, \ell, \varepsilon, c)}.$

Note that in every instance  $\mathcal I^*\sim\mathcal G^*(n, \ell, \varepsilon, c)$, every edge that does not belong to the random star $S$ has processing time 0 for one of its incident vertices. 
Thus, an optimal allocation depends only on the allocation of $S$, and if every $e_j \in S$ would be allocated to leaf $j$, then the total makespan would be at most 1.

\subsection{Allocation of a Fixed Star}
\label{sec:star}
In this subsection we consider a fixed set of parameters $n, \ell \in \mathbb N$, $\,\varepsilon \in (0,\frac{1}{K+1})$,  $\,c \in (0,\frac{1}{2}),$ and a fixed instance $\mathcal I^* $ from the support of ${\mathcal G^*(n, \ell, \varepsilon, c)}.$ Assume that $\mathcal I^*$ was obtained by choosing the multi-clique instance $\mathcal I,$ then selecting the root player $i\in N$ and the star $S=\{\,e_j \,|\, j\neq i\,\}$ of edges so that $e_j\in E_{i,j}.$

The instances $\mathcal I$ and $\mathcal I^*$  differ only in their processing times on $S\,$: for each $e_j \in S$ the processing time of player $i$ is $\,t_j=0\,$ in instance $\mathcal I,$ but it is $\,t^*_j=c\,$ in instance $\mathcal I^*.$ In both instances, $s_j \in (1{-}\varepsilon, 1)$ is the processing time that $j$ has for $e_j.$ We define two relevant subsets of the star edges for the given fixed allocation $A:\,$ Let
\[
	B^{\mathcal I}_{i,S} := \left\{e_j \in S\; ~\Big\vert~ \; 1{-}\varepsilon \,\leq\, \psi_{i,j}^{e_j, \mathcal I}(1-\varepsilon) \right\}.
\]
Note that if $e_j\in B^{\mathcal I}_{i,S},$ then the root player $i$ receives $e_j$ when $t_j$ is increased to $t_j':=1-2\varepsilon$ in instance $\mathcal I$ because for the critical value $1-2\varepsilon < \psi^{e_j,\mathcal I}_{i,j}(1-\varepsilon)\leq \psi^{e_j,\mathcal I}_{i,j}(s_j)$ holds. However, for the next lemma we will need that the root gets \emph{all} tasks from $B^{\mathcal I}_{i,S},$ if we set  $t_j'=1-2\varepsilon$ for \emph{all} $j\in B^{\mathcal I}_{i,S}$ at once.
A sufficient condition for this is that the star $S$ is a box.\footnote{The star $S$ being a box or not, concerns the allocation figure of the root player restricted to $S$ (cf. Figure~\ref{fig:box}); so being a box is an independent property of the setting of $t_j$ values and there cannot be a confusion between $\mathcal I$ and $\mathcal I^*$ but it is clearer to think of the instance $\mathcal I$ for this aspect.} Fortunately, a random star of root $i$ is a box with arbitrarily high probability, by the Box Theorem~\ref{thm:box-restated}.

We also define the following edge set w.r.t.  $\mathcal I^*$ and the given deterministic allocation $A:\,$
\[
	T^{\mathcal I^*} := \left\{e_j \in S\; ~\vert~\; A \text{ allocates } e_j \text{ to root } i \text{ for instance } \mathcal I^* \right\}.
\]

\begin{observation}
Since every $e_j \in S$ has a processing time of $c$ in instance $\mathcal I^*,$ it holds that $$\Mech{A}{\mathcal I^*} \ge c \cdot |T^{\mathcal I^*}|.$$
\end{observation}

Crucially, under the assumption that $S$ is an $\varepsilon$-box for $A,$ we can relate $|B^{\mathcal I}_{i,S}|$ and $|T^{\mathcal I^*}|.$ The next lemma is the core of the lower bound proof:

\begin{lemma}
	\label{lem:BTrelation}
	If $S$ is an $\varepsilon$-box for allocation algorithm $A$, then
	\[
		|T^{\mathcal I^*}| \ge \frac {(1-2\varepsilon) \cdot |B^{\mathcal I}_{i,S}| - c \cdot (n-1)} {1 - c}.
	\]
\end{lemma}

\begin{proof} For convenience, let $B := B^{\mathcal I}_{i,S}$ and $T := T^{\mathcal I^*} $. We consider the weakly monotone allocation $A$ to the root player $i,$ restricted to the $n-1$ tasks of the star $S,$ \emph{as a function of the values of the root $(t_j)_{j\neq i}$ only}. The input for every other task is fixed, as well as the costs $(s_j)_{j\neq i}$ of the other players on $S.$ It is known \cite{Vidali2009, ChristodoulouKK26} that depending only on $(t_j)_{j\neq i},$ a WMON allocation to player $i$ is determined by minimizing over linear expressions of the form $\sum_{e_j \in R} t_j + c_R$ for some constants $c_R,$ one for each subset $R\subseteq S.$ More precisely, there exist constants $c_R \in \mathbb R$  for every $R \subseteq S$ (determined by the payments to player $i$ for the sets $R$) so that $A$ allocates a subset $R \subseteq S$ to the root $i$ such that
	\[
		\sum_{e_j \in R} t_j + c_R
	\]
	is minimized (see also Subsection~\ref{sec:region-r_p}). W.l.o.g. we assume for simplicity that $c_S=0$ for the complete task set $S.$

	We will consider two different settings of $(t_j)_{j\neq i}.$ One of them is $t^*$ of the instance $\mathcal I^*;$ the other (called $t'$) will be one of the 'corners' of the box $S,$ namely we increase to $1-2\varepsilon$ every $t_j$ in those dimensions $j$ where the box is large: for all tasks of $B.$\footnote{ Geometrically, $\psi^{e_j}_{i,j}(s_j)$ is the length of the box in dimension $j$; and we know that for these tasks $\psi^{e_j}_{i,j}(s_j)\geq 1-\varepsilon.$ } In particular we  define $\mathcal I'$ by modifying $\mathcal I$ as follows:\footnote{Strictly speaking, $t'_j$ should be infinitesimally smaller than $1-2\varepsilon$ to avoid tie-breaking issues.} For every $e_j \in S$, let 
	\[
		t_j' := 
		\begin{cases} 
			1-2\varepsilon &\, \text{ if } e_j \in B,\\
			0 &\, \text{ if } e_j \in S\setminus B.
		\end{cases}
	\]
	The rest of the instance $\mathcal I'$ is the same as for $\mathcal I.$ Mind that the $c_R$ are the same constants for $\mathcal I,\, \mathcal I',$ and $\mathcal I^*,$ since they are independent of $(t_j)_{j\neq i}.$ They will help compare $|B|$ and $|T|.$ 
	
	We start by lower bounding every $c_R$ in terms of $|B|.$ By assumption, $S$ is an $\varepsilon$-box, meaning that $A$ allocates all edges of $S$ to root $i$ when every $t_j$ is set to at most $\psi^{e_j}_{i,j}(s_j)-\varepsilon$ at the same time. In $\mathcal I'$ this holds for every $e_j\in B,$ since $t'_j= 1-2\varepsilon = 1-\varepsilon -\varepsilon \leq \psi^{e_j}_{i,j}(1-\varepsilon)-\varepsilon\leq \psi^{e_j}_{i,j}(s_j)-\varepsilon,$ by definition of $B$ and monotonicity of $\psi_{i,j}.$ It also holds for every $e_j\in S\setminus B,$ because then $t'_j=0=\varepsilon-\varepsilon<\frac{1-\varepsilon}{K}-\varepsilon< \psi_{i,j}^{e_j}(1-\varepsilon)-\varepsilon\leq \psi_{i,j}^{e_j}(s_j)-\varepsilon.$ Here the first inequality holds due to $\varepsilon < \frac{1}{K+1};$ the second holds by Lemma~\ref{obs:alphasNew}; the third by monotonicity of $\psi_{i,j}.$
	
	This means that for input $\mathcal I'$ allocating the set $S$ to root $i$ minimizes $\,\sum_{e_j \in R} t_j' + c_R\,$ over all ${R\subseteq S},$ that is, for every $R\subseteq S$ we have, using $c_S=0$ that
	\begin{align*}
		&& \sum_{e_j \in S} t_j' + c_S &\le \sum_{e_j \in R} t_j' + c_R &&\\
		\quad &\Leftrightarrow& \quad \sum_{e_j \in S} t_j' - \sum_{e_j \in R} t_j' &\le c_R &&\\
		\quad &\Leftrightarrow& \quad \sum_{e_j \in B \setminus R} t_j' &\le c_R \;. &&
	\end{align*}
	Using the definition of $t_j'$ again, this is equivalent to
	\begin{equation}
		\label{eqn:affine-minimizer-offset}
		(1-2\varepsilon) \cdot |B \setminus R| \le c_R \;.
	\end{equation}
	
	Now consider the allocation for the instance $\mathcal I^*.$
	By definition, $T \subseteq S$ is the subset of edges that is allocated to root $i$ in this instance, and weak monotonicity implies

	\begin{align*}
		&& \sum_{e_j \in T} t_j^\star + c_T &\le \sum_{e_j \in S} t_j^\star + c_S &&\\
		&\Leftrightarrow& \quad c \cdot |T| + c_T &\le c \cdot (n-1) &&\\
		&\Leftrightarrow& \quad c_T &\le c \cdot \left( (n-1) - |T| \right).&&
	\end{align*}
	We used that $t_j^\star = c$ holds for all $e_j \in S$, and $|S| = n-1$. 
	Since Inequality~(\ref{eqn:affine-minimizer-offset}) holds particularly for $c_T,$ it follows
	\[
		(1-2\varepsilon) \cdot |B \setminus T| \le c \cdot \left( (n-1) - |T| \right).
	\]
	Observe that $(1-2\varepsilon) \cdot |B| - |T| \le (1-2\varepsilon) \cdot |B \setminus T|$, which implies
	\begin{align*}
		&& (1-2\varepsilon) \cdot |B| &\le |T| + c \cdot \left( (n-1) - |T| \right) &&\\
		&\Leftrightarrow& \quad (1-2\varepsilon) \cdot |B| - c \cdot (n-1) &\le |T| \cdot (1 - c),&&
	\end{align*}
	and the lemma follows.
	    \end{proof}

\subsection{Expectations for Random Stars}
\label{sec:random-star}

Pick  a random multi-clique instance $\mathcal I \sim \mathcal G(n, \ell, \varepsilon);$ then choose a random root player $i$ and  a random  star $S$ of incident edges $e_j$ out of those with processing time $t_j=0$ for the player $i$ (from the sets $E_{ij}$).
We lower bound the expectation  of $|B|$  the number of edges of the star with 'large' value of $\psi_{i,j}^{e_j}$ (over the random choice of $\mathcal I,\, i,$ and $S$). %
Intuitively it seems clear that $\psi(1-\varepsilon)\geq 1-\varepsilon$ for at least half of all possible $\psi_{i,j}^{e}$ functions (for given $A$ and $\mathcal I$), since if the costs of an edge are set to $1-\varepsilon$ for both of its endpoints, it must be allocated to one of them. However, it is not obvious why $\mathbb E[\,|B|\,]\geq \frac{n-1}{2}$ should hold for a random star, because for given center $i$ the star edges $e_j$ are picked from the $E_{i,j},$ precisely the complementary edge sets to $E_{j,i}.$ Therefore the lower bound $\frac{n-1}{2}$ holds only in expectation taken (also) over $\mathcal I \sim \mathcal G(n, \ell, \varepsilon),$ in particular over the random directions of $0 - 1$ costs over the edges. Recall that $Q_{\{i,j\}}$ denotes the set of all parallel edges between $i$ and $j.$ The proof uses the following crucial observation from \cite{ChristodoulouKK26}:

\begin{proposition}
\label{prop:indepParallel}
If the mechanism has finite approximation ratio, then w.l.o.g. $\varepsilon$ can be chosen so (in fact, a randomly perturbed $\varepsilon$ will do), that the critical value $\psi_{i,j}^{e_j,\mathcal I}(1-\varepsilon)$ does not depend on the running times of edge $e_j,$ nor does it depend on the running times on any parallel edge $e\in Q_{\{i,j\}}.$ 
\end{proposition}

\begin{proof} Trivially, by definition the value $\psi_{i,j}^{e_j,\mathcal I}(1-\varepsilon)$ does not depend on the actual costs on edge $e_j$ in $\mathcal I.$

For the independence from the costs on parallel edges (so called \emph{siblings} of $e_j$), the argument is essentially the same as for Theorem~\ref{thm:sibling-independence} in Section~\ref{sec:facts-about-two}: Since the instance $\mathcal I$ is trivial (i.e., every task has at least one zero entry), if we fix this trivial input on all tasks except for a pair $(e, e')$ of sibling tasks, the mechanism restricted to these two tasks -- the so called \emph{$(e, e')$-slice mechanism} -- is a truthful mechanism for two players and two tasks, that must have finite approximation ratio, otherwise the whole mechanism would have infinite ratio. By the characterization of $2\times 2$ mechanisms (see Theorem~\ref{theo:addchar} in Section~\ref{sec:facts-about-two}) the slice mechanism is essentially a \emph{relaxed affine minimizer} or a \emph{relaxed task-independent} mechanism. If it were a (non task-independent) relaxed affine minimizer (or a \emph{constant-} or  \emph{bundling-mechanism}), then the approximation factor would be infinite (see Lemma~\ref{prop:quasi}). 

Given that, consequently, the mechanism is relaxed task-independent, the only possibility for the critical value $\psi^e(1-\varepsilon)$ to change by changing a running time on the sibling task $e',$ would be if the $\psi^{e}$ had a jump-discontinuity in the point $1-\varepsilon$ (see, e.g., \cite{ChristodoulouKoutsoupiasKovacs2020Submodular}). Recall that the $\psi^e$ functions are increasing, and thus have only countably many discontinuity points. For each given \emph{fixed} set $\mathcal I_s$ of perturbed costs, we can select a perturbed $\varepsilon$ so that for every setting of $\mathcal I$ (of the $0-1$ directions) on the finitely many edges, the $1-\varepsilon$ is no discontinuity point of  $\psi^e$ for any edge $e\in E.$ (Note that the instance $\mathcal I$ itself does not depend on $\varepsilon,$ nor do the \emph{functions} $\psi^e.$)
\end{proof}

\begin{lemma} Select a random instance $\mathcal I \sim \mathcal G(n, \ell, \varepsilon),$ and then select uniformly at random a root $i\in N$ and finally pick a star $S$ by choosing uniformly an edge $e_j\in E_{ij}$ for each $j\neq i.$ Denote $B^{\mathcal I}_{i,S} := \left\{e_j \in S ~\Big\vert~ 1{-}\varepsilon \leq  \psi_{i,j}^{e_j,\mathcal I}(1-\varepsilon) \right\}.$ For the expected size of $B$ holds that
	 $$\mathbb E[\,|B^{\mathcal I}_{i,S}| \, ] \geq \frac {n-1}{2}.$$
\end{lemma}
\begin{proof}

For a given root $i$ let us call a star $S=\{\,e_{j,S}\,|\,j\in N,\; j\neq i\,\}$  \emph{proper}, if $e_{j,S}\in E_{ij}^{\mathcal I}$ for every $j\neq i.$  We denote by $\mathcal S_i^{\mathcal I}$  the set of all proper stars with root $i,$ as well as (abusing notation) the uniform distribution over stars in $\mathcal S_i^{\mathcal I}.$  
 Similarly, we write $i\sim N$ for a uniformly random selection of a player $i.$ We use $r:=\ell/2.$ Recall that $|E_{i,j}|=r$ holds for every $\mathcal I, i,$ and $j.$ \footnote{We will omit  $\mathcal I$  from the notation of $\mathcal S_i^{\mathcal I}, B^{\mathcal I}_{i,S},\, \psi_{i,j}^{e_j,\mathcal I},\,E_{i,j}^{\mathcal I},...$ etc. whenever $\mathcal I$ is clear from the context.}

\emph{Consider first a fixed multi-clique instance $\mathcal I$ from the support of $\mathcal G(n, \ell, \varepsilon).$} For this $\mathcal I$ we will express $\mathbb E_{i\sim N,\, S\sim \mathcal S_i}[\,|B_{i,S}|\,]$ with the help of indicator variables. For any edge $e\in Q_{\{i,j\}}$  let the variable $x^e_{ij}$ indicate whether $\psi_{i,j}^{e}(1-\varepsilon)\geq 1-\varepsilon\,$ holds, i.e., $x_{ij}^e=1$ if $\psi_{i,j}^{e}(1-\varepsilon)\geq 1-\varepsilon,\,$ and $x_{ij}^e=0$ otherwise. By definition $|B_{i,S}|= \sum_{j\,|\,j\neq i} x^{e_{j,S}}_{ij}$ holds. Observe that, since one of $i$ or $j$ gets the edge if their processing times are set to $(1-\varepsilon,1-\varepsilon),$ for the sum $x_{ij}^e+x_{ji}^e\geq 1$ holds (in fact, $x_{ij}^e+x_{ji}^e= 1$ unless $\psi(1-\varepsilon)=1-\varepsilon$ for both players).    

We also introduce another type of indicator:  For given $\mathcal I$ and arbitrary edge $e\in Q_{\{i,j\}}$ let $y_{ij}^e=1$ if $e\in E_{i,j},$ and $y_{ij}^e=0$ if $e\in E_{j,i}.$ Now  we have

	\begin{align*}
		\mathbb E_{i\sim N,\, S\sim \mathcal S_i}[\,|B_{i,S}|\,] 
		&=\sum_{i\in N} \frac{1}{n}\sum_{S\in \mathcal S_i}\frac{1}{r^{n-1}}|B_{i,S}|\\
		&=\sum_{i\in N} \frac{1}{n}\sum_{S\in \mathcal S_i}\frac{1}{r^{n-1}}\sum_{j\,:\,j\neq i} x^{e_{j,S}}_{ij}\\
		&= \frac{1}{n\cdot r^{n-1}}\sum_{i\in N}\sum_{j\,:\,j\neq i}\sum_{e\in E_{i,j}} x^{e}_{ij}\cdot\left |\{S\in \mathcal S_i \, |\, e \in S\} \right |\\	
		&= \frac{1}{n\cdot r^{n-1}}\sum_{i\in N}\sum_{j\,:\,j\neq i}\sum_{e\in E_{i,j}} x^{e}_{ij}\cdot r^{n-2}\\
		&= \frac{1}{n\cdot r}\sum_{i\in N}\sum_{j\,:\,j\neq i}\sum_{e\in E_{i,j}} x^{e}_{ij}\\ 
		&= \frac{1}{n\cdot r}\sum_{i\in N}\sum_{j\,:\,j\neq i}\sum_{e\in Q_{\{i,j\}}} y^e_{ij} x^{e}_{ij}.\\ 
	\end{align*}
	
It remains to bound the expectation of the last expression  w.r.t. the random choice of the multi-clique instance $\mathcal I\sim \mathcal G(n, \ell, \varepsilon).$ Recall that we consider each instance  as a pair  $\mathcal I_s\times \mathcal I$ of independent random choices as follows. We fixed $\mathcal I_s$ that assigned (randomly perturbed) values from $(1-\varepsilon,1)$ to the edges, one such value per edge;  $\mathcal I$ stands for the random selection  of the bipartition $(E_{i,j}, E_{j,i})$ with $|E_{i,j}|=\ell/2,$ for each pair of players $\{i,j\},$ deciding which one of them gets the $0$ running time and which the (perturbed) $s^e$ for each edge of $e\in Q_{\{i,j\}}.$ We assumed that $\mathcal I_s$ is selected and fixed \emph{before} $\mathcal I$ is picked. The next argument treats expectation w.r.t. $\mathcal I\sim \mathcal G(n, \ell, \varepsilon),$ and holds for every fixed $\mathcal I_s.$ 

 In the last sum of the above calculation each unordered pair $\{i,j\}$ of vertices occurs twice. We calculate the expectation over $\mathcal I$ for each unordered pair $\{i,j\}$ separately: We  divide $\mathcal I$  to independently selected parts $\mathcal I=\mathcal I_{\{i,j\}}\times \mathcal I_{-\{i,j\}},$ where
 $\mathcal I_{\{i,j\}}$ denotes the setting of orientations of $(0,s^e)$ on the $\ell$ parallel edges of $e\in Q_{\{i,j\}}.$ The orientation on all \emph{other} edges will be denoted $\mathcal I_{-\{i,j\}}.$ Now, with respect to the distribution of $\mathcal I_{\{i,j\}}\times \mathcal I_{-\{i,j\}}$ the $x^e_{ij}$ and $y^e_{ij}$ are \emph{random} variables. The key observation is that the $x^e_{ij}$ do not depend on $\mathcal I_{\{i,j\}}$ according to Proposition~\ref{prop:indepParallel}. The $x^e_{ij}$ do depend on $\mathcal I_{-\{i,j\}},$ because the values on edges other than $Q_{\{i,j\}}$ might influence the allocation functions $\psi_{i,j}$ on $Q_{\{i,j\}}.$ In fact, we can fix  $\mathcal I_{-\{i,j\}};$ for the given $A$ this determines $\psi^e_{i,j}$ (and thus $x^e_{ij}$) on each edge of $Q_{\{i,j\}},$ and we may choose $\mathcal I_{\{i,j\}}$ only afterwards. For the $y^e_{ij}$ it is the other way around: obviously these depend only on $\mathcal I_{\{i,j\}}$ but not on $\mathcal I_{-\{i,j\}}.$ By definition, for $\mathcal I_{\{i,j\}}$ all possible selections of $r$ edges for $E_{ij}$ have equal probability and thus $\mathbb E_{\mathcal I}[y^e_{ij}]=\mathbb E_{\mathcal I_{\{i,j\}}}[y^e_{ij}]=1/2.$  Since $\mathcal I_{\{i,j\}}$ and $\mathcal I_{-\{i,j\}}$ are independent,  $x^e_{ij}$ and $y^e_{ij}$ are independent variables.

  Let $[N]^2$ denote the set of all unordered pairs of players, then 
  
  \begin{align*}
		\mathbb E_{\mathcal I}\left [\sum_{i\in N}\sum_{j\,:\,j\neq i}\sum_{e\in Q_{\{i,j\}}} y^e_{ij} x^{e}_{ij}\right ] 
		&=\mathbb E_{\mathcal I}\left [\sum_{\{i,j\}\in [N]^2}\sum_{e\in Q_{\{i,j\}}} (y^e_{ij} x^{e}_{ij}+ y^e_{ji} x^{e}_{ji})\right ]\\
		&=\sum_{\{i,j\}\in [N]^2}\mathbb E_{\mathcal I} \left [\sum_{e\in Q_{\{i,j\}}} (y^e_{ij} x^{e}_{ij}+ y^e_{ji} x^{e}_{ji} )\right]\\
		&=\sum_{\{i,j\}\in [N]^2} \sum_{e\in Q_{\{i,j\}}} (\mathbb E_{\mathcal I }[y^e_{ij}]\cdot \mathbb E_{\mathcal I }[x^{e}_{ij}]+ \mathbb E_{\mathcal I }[y^e_{ji}]\cdot \mathbb E_{\mathcal I }[x^{e}_{ji}] )\\
		&=\sum_{\{i,j\}\in [N]^2} \sum_{e\in Q_{\{i,j\}}} \left (\frac{1}{2}\cdot \mathbb E_{\mathcal I }[x^{e}_{ij}]+ \frac{1}{2}\cdot \mathbb E_{\mathcal I }[x^{e}_{ji}] \right)\\
		&=\frac{1}{2}\sum_{\{i,j\}\in [N]^2} \sum_{e\in Q_{\{i,j\}}} \mathbb E_{\mathcal I }[x^{e}_{ij}+ x^{e}_{ji}] \\
		&\geq \frac{1}{2}\sum_{\{i,j\}\in [N]^2} \sum_{e\in Q_{\{i,j\}}} \mathbb E_{\mathcal I } [1] = \frac{1}{2}\cdot\frac{n(n-1)}{2}\cdot\ell\cdot 1=\frac{n(n-1)r}{2}.\\
			\end{align*}
			The third equality uses the independence of $x_{ij}$ and $y_{ij};$ the fourth uses that $\mathbb E_{\mathcal I}[y^e_{ij}]=1/2;$ finally the inequality holds due to $x_{ij}^e+x_{ji}^e\geq 1.$
			Plugging this result into the formula for $\mathbb E_{i\sim N,\, S\sim \mathcal S_i}[\,|B_{i,S}|\,], $ we obtain 
			$$\mathbb E_{\mathcal I\sim \mathcal G, \,i\sim N,\, S\sim \mathcal S_i}[\,|B^{\mathcal I}_{i,S}|\,]\geq \frac{1}{n\cdot r}\frac{n(n-1)r}{2}=\frac{(n-1)}{2}. $$\end{proof}

Since the star $S$ is an $\varepsilon$-box with arbitrarily high probability $1-\Delta,$ we obtain expectation on condition that $S$ is a box, arbitrarily close to $\frac{n-1}{2},$ as formalized in the corollary: 

\begin{corollary}\label{cor:BoundB} Let $\mathcal I$ be a random instance $\mathcal I \sim \mathcal G(n, \ell, \varepsilon).$  Select uniformly at random a root $i\in N$ and then a star $S$ of edges $e_j\in E_{ij}$ and let $B := \left\{e_j \in S ~\Big\vert~ 1{-}\varepsilon \leq  \psi_{i,j}^{e_j,\mathcal I}(1-\varepsilon) \right\}.$ For the conditional expectation of the size of $B$ holds that
	 $$\mathbb E_{\mathcal I\sim \mathcal G,\, i\sim N, S\sim \mathcal S_i}\left[|B| \;\Big\vert\; S \textnormal{ is an $\varepsilon$-box}\, \right] \geq \frac {n-1} 2\cdot (1-2\Delta),$$
	 given that the probability of the star $S$ of not being an $\varepsilon$-box is at most $\Delta.$ 
\end{corollary}

\begin{proof}  We simply write $\mathbb E[\,\,] $ instead of  $\mathbb E_{\mathcal I\sim \mathcal G,\, i\sim N, S\sim \mathcal S_i}[\,\,].$ Assume for contradiction that \linebreak $\mathbb E[\;|B|\;|\,S \textnormal{ is an $\varepsilon$-box} ]<\frac {n-1} 2\cdot (1-2\Delta).$ Then 

\begin{align*}\frac{n-1}{2}\leq \mathbb E[\,|B|\,]
&=\mathbb E[\;|B|\;|\;S \textnormal{ is an $\varepsilon$-box} ]\cdot \mathbb P[S \textnormal{ is an $\varepsilon$-box }]\\
&+\mathbb E[\;|B|\;|\;S \textnormal{ is not an $\varepsilon$-box} ]\cdot \mathbb P[S \textnormal{ is not an $\varepsilon$-box }]\\
&<\frac {n-1} 2\cdot (1-2\Delta)\cdot 1 + (n-1)\cdot \Delta=\frac{n-1}{2},\\					\end{align*}
a contradiction. For the inequality we used that $\mathbb E[\,|B|\,|\,S \textnormal{ is an $\varepsilon$-box} ]<\frac {n-1} 2\cdot (1-2\Delta)$ by the indirect assumption; $\mathbb P[S \textnormal{ is an $\varepsilon$-box }]\leq 1;$ $\quad \mathbb E[\,|B|\,|\,S \textnormal{ is not an $\varepsilon$-box} ]\leq n-1$ since $|B|\leq | S|= n-1,$ and  $\mathbb P[S \textnormal{ is not an $\varepsilon$-box }]\leq \Delta$ (for an appropriate $\Delta$ to be specified later) according to Theorem~\ref{thm:box-restated}. (The latter probability holds w.r.t. $S\sim \mathcal S_i$ for every fixed $\mathcal I$ in the support of $\mathcal G$, and every $i\in N.$) \end{proof}

Recall that for any parameters $c\in (0, \frac{1}{2}),$ $\varepsilon\in (0,\frac{1}{K+1}),$ and  $\ell \in \mathbb N$ a random instance $\mathcal I^* \sim \mathcal G^*(n, \ell, \varepsilon, c)$ is generated by picking a random clique-instance $\mathcal I \sim \mathcal G(n, \ell, \varepsilon),$ then selecting a random root player $i\sim N,$ a random star $S$ of edges from the $E_{i,j}$ sets, and finally increasing from $t_j=0$ to $t^*_j=c$ the running times of the tasks/edges $e_j$ of the star $S$ on the root machine $i.$  

Let $A$ be a  deterministic truthful allocation algorithm $A \in \mathcal A$ with approximation ratio less than $K.$ 
Combining Corollary~\ref{cor:BoundB} with Lemma~\ref{lem:BTrelation}, yields the following lower bound for the expected makespan of $A$: 

\begin{lemma}\label{lem:LBMakespan} Let $\Delta=\Delta(n,\ell,K,\varepsilon)$ denote the upper bound on the probability\footnote{$\Delta(n,\ell,K,\varepsilon)=\left(\frac{5K4^{n-1}}{\varepsilon}\right)^{n-1}\frac{1}{\sqrt{\ell}}$} that a random star in a multi-clique $\mathcal I$ is not an $\varepsilon$-box as given by the Box Theorem~\ref{thm:box-restated}. Then the expected makespan of a random instance $\mathcal I^* \sim \mathcal G^*(n, \ell, \varepsilon, c)$ is lower bounded  by

	$$
		\mathbb E_{\mathcal I^*\sim \mathcal G^*}[\Mech{A}{\mathcal I^*}] 
		\ge \frac{c\,(n-1)}{2(1-c)}\cdot\,	\left[1-2c-2(\varepsilon+\Delta)\right]\,\left(1-\Delta\right);
	$$
	 this lower bound goes to $
		 \; \frac{c\,(1-2c)\,}{2\,(1-c)}\cdot (n-1)
	$
	when $\varepsilon, \,\Delta \,\rightarrow \, 0;$ that is, for $\varepsilon \, \rightarrow \, 0, $ and (for given $\varepsilon$) $\ell \,\rightarrow \,\infty.$

\end{lemma}

\begin{proof} Let $\mathcal I^*$ be an instance in the support of $G^*(n, \ell, \varepsilon, c),$ let $\mathcal I$ be the corresponding multi-clique instance (i.e., before increasing the $t_j=0$ values on a star), and $S$ be the star with root $i$ selected for $\mathcal I^*.$ Assume that (for $A$ and $\mathcal I$) the star $S$ is an $\varepsilon$-box. Then by the definition of $T^{\mathcal I^*}$ and by Lemma~\ref{lem:BTrelation} we have 
\[ \Mech{A}{\mathcal I^*} \ge c \cdot |T^{\mathcal I^*}|\geq \, c\cdot
		\frac {(1-2\varepsilon) \cdot |B^{\mathcal I}_{i,S}| - c \cdot (n-1)} {1 - c} .
	\]

	Since this lower bound on the makespan holds for every single  instance $\mathcal I^*,$ therefore it also holds in expectation, on condition that $S$ is an $\varepsilon$-box:
	
	\begin{align*}
	\mathbb E_{\mathcal I^*}[\Mech{A}{\mathcal I^*}] 
	&\ge  \mathbb E_{\mathcal I^*}\left[ c \cdot | T^{\mathcal I^*} | \; \right]\\
	&\ge  \mathbb E_{\mathcal I^*}\left[ c \cdot | T^{\mathcal I^*} | \;\vert \; S \text{ is an $\varepsilon$-box}\right]\,\cdot\,\mathbb P\left [S \text{ is an $\varepsilon$-box}\right]\\
	&\ge  \frac{c}{1-c}\cdot\,\left((1-2\varepsilon)\cdot \mathbb E_{\mathcal I^*}\left[ \, | B^{\mathcal I}_{i,S} | \, \;\vert \; S \text{ is an $\varepsilon$-box}\right] - c\cdot (n-1)\right)\,\cdot\,(1-\Delta)\\
	&\ge  \frac{c}{1-c}\cdot\,\left((1-2\varepsilon)\cdot \frac {n-1} 2\cdot (1-2\Delta) - c\cdot (n-1)\right)\,\cdot\,(1-\Delta)\\
	&\ge  \frac{c (n-1)}{2(1-c)}\cdot \,\left((1-2\varepsilon)\cdot  (1-2\Delta) - 2c\right)\,\cdot\,(1-\Delta)\\
	&\ge \frac{c\,(n-1)}{2(1-c)}\cdot\,	\left(1-2c-2(\varepsilon+\Delta)\right)\,\left(1-\Delta\right).\\
	\end{align*}
	The third inequality uses Lemma~\ref{lem:BTrelation} and the fourth uses Corollary~\ref{cor:BoundB}.
	\end{proof}

\begin{theorem}\label{thm:LB} Let  $c\in (0, \frac{1}{2}).$ For the distribution $\mathcal I^* \sim \mathcal G^*(n, \ell, \varepsilon, c)$  for arbitrary deterministic weakly monotone allocation algorithm $A\in \mathcal A_K$ holds that for arbitrary $\delta>0$

$$\frac{\mathbb E_{\mathcal I^* \sim \mathcal G^*}[\Mech{A}{\mathcal I^*}]}{\mathbb E_{\mathcal I^* \sim \mathcal G^*}[\Opt{\mathcal I^*}]}\; > \;\frac{c\,(1-2c)\,}{2\,(1-c)}\cdot (n-1)-\delta \,$$
for small enough $\varepsilon >0$ and large enough multiplicity of parallel tasks $\ell(\varepsilon) \in \mathbb N.$ 

\end{theorem}

\begin{proof} $${\mathbb E_{\mathcal I^* \sim \mathcal G^*}[\Mech{A}{\mathcal I^*}]}\, > \,\frac{c\,(1-2c)\,}{2\,(1-c)}\cdot (n-1)-\delta\,$$ holds for $\varepsilon$ small enough and  $\ell=\ell(\varepsilon)$ large enough, by the previous lemma.

Consider now $\mathbb E_{\mathcal I^* \sim \mathcal G^*}[\Opt{\mathcal I^*}].$ For every instance $\mathcal I^*$ in the support of $\mathcal G^*$ holds that $\Opt{\mathcal I^*}\leq \max_{j| j\neq i} s_j< 1,$ since the edges of the star $S$ can be allocated to the respective leaf vertices, and all other tasks are trivial. Thus $\frac{\mathbb E_{\mathcal I^*}[\Mech{A}{\mathcal I^*}]}{\mathbb E_{\mathcal I^* }[\Opt{\mathcal I^*}]}> \frac{\mathbb E_{\mathcal I^*}[\Mech{A}{\mathcal I^*}]}{1}> \,\frac{c\,(1-2c)\,}{2\,(1-c)}\cdot (n-1)-\delta\,$ holds for the ratio as well. 
\end{proof}
Using the Yao Principle,\footnote{The Yao principle is based on the minimax Theorem for zero-sum games (between algorithm designer and the adversary). The entries of the game matrix  are here not the approximation ratios $\rho(A,\mathcal I),$ but for any fixed $\beta,$ the $\Mech{A}{\mathcal I}-\beta\cdot\Opt{\mathcal I}.$} we obtain the main lower-bound result:

\begin{theorem}\label{thm:mainLB} Let $n$ denote the number of machines and let $c\in (0, \frac{1}{2}).$  Then $\frac{c\,(1-2c)\,}{2\,(1-c)}\cdot (n-1)$ is a lower bound on the approximation factor $\rho(R)$ of any (discrete) randomized universally truthful mechanism $R.$ 
By setting  $c=1/4$ or $c=1/3,$ we obtain the lower bound $$\rho(R)\geq \frac{(n-1)}{12}\approx 0.0833\cdot (n-1).$$ 
The maximum value of $\frac{c\,(1-2c)\,}{2\,(1-c)}$ is $(3/2-\sqrt{2})\approx 0.0858$ achieved for $c=1-\frac{\sqrt{2}}{2}.$ The latter yields the bound $$\rho(R)\geq (3/2-\sqrt{2})(n-1)\approx 0.0858\cdot(n-1).$$

\end{theorem}

\begin{proof} For arbitrary given $K\geq n+1$ and for every deterministic mechanism $A\in \mathcal A_K,$ we obtained the lower bound  of Theorem~\ref{thm:LB} for the given distribution over multi-clique instances.  Applying Yao's Principle~\cite{Yao77,MotwaniR95}, every randomized universally truthful mechanism $R\in \mathcal R_K$ for unrelated machine scheduling has expected approximation ratio $\rho(R)\geq \frac{c\,(1-2c)\,}{2\,(1-c)}\cdot (n-1).$%

By Proposition~\ref{prop:discrete}, this further implies the same approximation lower bound for every universally truthful mechanism that is a discrete distribution over (countably many) deterministic truthful mechanisms. 
\end{proof}

	\section{Upper Bound}
\label{sec:upperbound}

In this section, we complement our negative result, by proposing a
universally truthful mechanism, which we call
\textsc{Exp-Bounded-Square}, that achieves an expected approximation
ratio of at most $(n+1)/2 + O(\sqrt{n \ln n}) $
(Corollary~\ref{cor:upper-bound}). This is an asymptotic improvement
over the linear coefficient of the previous best bound of $0.837n$
\cite{LuYu2008Improved} and it asymptotically matches the $(n+1)/2$
bound for
fractional~\cite{ChristodoulouKoutsoupiasKovacs2010Fractional} as well as the
$(n+5)/2$ bound for truthful-in-expectation
mechanisms~\cite{LuYu2008Randomized}.  Our bound also asymptotically
matches the best possible bound that can be achieved by
\emph{task-independent}
mechanisms~\cite{ChristodoulouKoutsoupiasKovacs2010Fractional,LuYu2008Randomized}.

In Section~\ref{sec:exp-bounded-square}, we formally define
\textsc{Exp-Bounded-Square} and derive the task assignment
probabilities as a function of the machines' bids.  Next,
Section~\ref{sec:universal-truthfulness} establishes the universal
truthfulness of the mechanism, and Section~\ref{sec:approximation}
presents the approximation guarantee.

\subsection{The \textsc{Exp-Bounded-Square} mechanism}
\label{sec:exp-bounded-square}
\textsc{Exp-Bounded-Square}, presented in
Algorithm~\ref{alg:exp-bounded-square}, is an adaptation of the
truthful-in-expectation \textsc{Bounded-Square} mechanism by Lu and
Yu~\cite{LuYu2008Randomized}.

\begin{algorithm}[ht]
\DontPrintSemicolon
\caption{\textsc{Exp-Bounded-Square}}\label{alg:exp-bounded-square}
\KwIn{Non-negative processing times $t = (t_{ij})_{i \in N, j \in M}$.}
\KwOut{Random allocation of tasks (set $M$) to machines (set $N$).}
\For{$j = 1, \dots, m$}{
	\For{$i = 1, \dots, n$} {
		Draw $E_{ij} \sim \text{Exp}(1)$ independently.\;
	}
	$t_j^{\min} \gets \min_{i \in N} \{ t_{ij} \}$ \;
	$N_j \gets \{i \in N ~\colon~ t_{ij} \le 2 \cdot t_j^{\min} \}$. \;
	
	\eIf{$t_j^{\min} = 0$}{
		Assign $j$ to a machine $i^*$ drawn uniformly at random from $N_j$.\;
	} {
		Assign $j$ to machine $i^* = \arg \min_{i \in N_j} \left\{ E_{ij} \cdot (t_{ij})^2 \right\}$.\;
	}
}
\end{algorithm}

The mechanism processes tasks independently using a two-stage randomized assignment rule. 
First, it restricts the allocation to an eligible set of machines $N_j$ whose processing times are within a factor of $2$ of the minimum, $t_j^{\min} = \min_{i} t_{ij}$. 
This pruning step prevents the mechanism from assigning a task to a highly inefficient machine. 
Then, if $t_j^{\min} = 0$, the task is allocated uniformly at random among those machines. Otherwise, it is assigned to the machine that minimizes the \emph{exponentially scaled quadratic} processing time, defined as $i^* = \arg \min_{i \in N_j} \{ E_{ij} \cdot (t_{ij})^2 \}$, where $E_{ij} \sim \text{Exp}(1)$ are independent random weights drawn from the exponential distribution with rate parameter 1. 
This quadratic scaling biases the allocation toward faster machines and was first used for the fractional case in~\cite{ChristodoulouKoutsoupiasKovacs2010Fractional}. The independent random weights drawn from the exponential distribution
serve to guarantee universal truthfulness. Because the weights are drawn from a continuous probability distribution, the algorithm has continuous support. We assume an arbitrary but fixed tie-breaking rule that is independent of the input $t$.

We first show in the next lemma that the assignment of each task follows
the exact same probability distribution as the fractional solution
computed by \textsc{Bounded-Square}~\cite{LuYu2008Randomized}. Given any
instance $t$, for every task $j \in M$, let
$t_j^{\min} := \min_{i \in N} \{ t_{ij} \}$ denote the minimum
processing time for task $j$, and let
\[
	N_j := \{i \in N \;\colon\; t_{ij} \le 2 \cdot t_j^{\min} \}
\]
denote the set of \emph{eligible machines} for task $j$, as described in Algorithm~\ref{alg:exp-bounded-square}.

\begin{lemma}
	\label{lem:probabilities-exp-bounded-square}
	Given any instance $t$, \textsc{Exp-Bounded-Square} allocates every task $j \in M$ independently to an eligible machine $i \in N_j$ with probability
	\[
		\begin{cases}
			\frac {1/(t_{ij})^2} {\sum_{s \in N_j} 1/(t_{sj})^2} &, ~\text{if}~ t_j^{\min} > 0 ~,\\
			\frac {1} {|N_j|} &, ~\text{else}~,
		\end{cases}
	\]
	and to a non-eligible machine $i \notin N_j$ with probability 0.
\end{lemma}
\begin{proof}
	It follows directly from Algorithm~\ref{alg:exp-bounded-square} that the allocation is task-independent, and every task $j \in M$ with $t_j^{\min} = 0$ is allocated to an eligible machine uniformly at random.
	Therefore, we consider any task $j \in M$ with $t_j^{\min} > 0$.
	Let $E_{ij} \sim \text{Exp}(1)$ (i.i.d.) and $X_i := E_{ij} \cdot (t_{ij})^2$ for every $i \in N_j$.
	By description, Algorithm~\ref{alg:exp-bounded-square} allocates $j$ to machine $i^* = \arg\min_{i \in N_j} \{X_i\}$.
	It follows 
	\[
		X_i \overset{\text{i.i.d.}}\sim \text{Exp}(1/(t_{ij})^2)
	\]
	directly from the definition of the exponential distribution, and subsequently
	\[
		\mathbb P \left( X_{i^*} = \min_{i \in N_j} \{	X_i \} \right) = \frac {1/(t_{i^*j})^2} {\sum_{i \in N_j} 1/(t_{ij})^2} \enspace.
	\]
\end{proof}

\subsection{Universal Truthfulness}
\label{sec:universal-truthfulness}

Next, we establish that the mechanism is indeed universally truthful.
Although not needed for the proof, we discuss a possible payment scheme in \Cref{sec:payment-scheme-upperbound} for comprehensiveness of the mechanism.

\begin{lemma}
	\textsc{Exp-Bounded-Square} is a universally truthful mechanism.
\end{lemma}
\begin{proof}
	We consider an arbitrary but fixed instance, and an arbitrary but fixed allocation rule.
	Thus, for each machine $i \in N$ and each task $j \in M$, the random draw of $E_{ij}$ is fixed to an arbitrary realization.
	Moreover, due to task-independence of the mechanism, it suffices to consider the allocation of a fixed task $j \in M$.
	The reported processing times for other tasks do not influence the allocation of $j$.

	Suppose $j$ is allocated to machine $i^* \in N$ under some reported
	processing time $t_{i^*j} > 0$.  Then it holds
	$i^* = \arg \min_{i \in N_j} \left\{ E_{ij} \cdot (t_{ij})^2
	\right\}$.  When reporting $t_{i^*j}' < t_{i^*j}$ while all
	other processing times $t_{ij}$ reported by machines
	$i \neq i^*$ remain fixed, the new set of \emph{eligible
	  machines} is a (weak) subset of $N_j$, since $t_j^{\min}$
	might be decreased.  However $i^*$ is still an eligible
	machine, and therefore $E_{i^*j} \cdot (t_{i^*j})^2$ is still
	the minimum value that can be achieved with eligible machines.
	Thus task $j$ will still be allocated to $i^*$, and the
	allocation is weakly monotone.
	As stated in~\Cref{sec:prelim}, this implies truthfulness.
\end{proof}

\subsection{Approximation}
\label{sec:approximation}

Here we show an upper bound on the expected makespan of \textsc{Exp-Bounded-Square}.
For any given instance $t$, let $R_i(t)$ denote the random set of tasks that is allocated to machine $i \in N$ by \textsc{Bounded-Square}, and let $\theta_i(t) := \sum_{j \in R_i(t)} t_{ij}$ denote the total (random) processing time of machine $i$.
Moreover, let 
\[
	\Mech{\textsc{EBS}}{t} := \max_{i \in N} \theta_i(t)
\] 
denote the makespan of the random allocation obtained by \textsc{Exp-Bounded-Square} under $t$, and let
\[
	\mu(t) := \frac {n+1} 2 \cdot \Opt{t} \enspace,
\]
where $\Opt{t}$ is the optimal makespan under $t$.

We formally denote two simple observations.
Every task $j \in M$ has a processing time of at least $t_j^{\min} = \min_{i \in N} \{ t_{ij} \}$ on every machine.
Hence $\max_{j \in M} t_j^{\min} \le \Opt{t}$.
Consequently, even when every task would be allocated to a machine that has minimum processing time for it, the average total processing time of the machines would be $\frac 1 n \cdot \sum_{j \in M} t_j^{\min}$, and the maximum total processing time in an optimal allocation cannot be better, which implies $\sum_{j \in M} t_j^{\min} \le n \cdot \Opt{t}$. 
We may also assume that there exists a task $j \in M$ with $t_j^{\min} > 0$.
Otherwise, \textsc{Exp-Bounded-Square} allocates every task $j \in M$ to a machine $i \in N$ that has processing time $t_{ij} = 0$ for it, and it holds $\Mech{\textsc{EBS}}{t} = 0$ with probability 1.
In this case, the allocation would be optimal.
	
\begin{observation}
	\label{observation:opt-bounds}
	It holds $0 < \max_{j \in M} t_j^{\min} \le \Opt{t}$
	and $0 < \sum_{j \in M} t_j^{\min} \le n \cdot \Opt{t}$.
\end{observation}

Lu and Yu \cite{LuYu2008Randomized} gave a bound on the maximum expected processing time of the machines, which also holds for our algorithm due to identical (fractional) probabilities of the allocation.
For completeness, we provide a formal proof in~\Cref{sec:expected-load-lu-yu-upperbound}.
Note that this result does not directly imply a bound for the {\em expected makespan} $\mathbb E[\Mech{\textsc{EBS}}{t}]$, which we aim to obtain.
\begin{lemma}[\cite{LuYu2008Randomized}, Lemma 3]
	\label{lem:lu_yu_max_expected_load}
	For any instance $t$, it holds
	\[
		\mathbb E[\theta_i(t)] \le \mu(t) = \left(\frac{n+1}{2}\right)\cdot \Opt{t}
	\] 
	for every machine $i \in N$.
\end{lemma}

In the following lemma, we show that the completion time of a machine deviates above its expected value with exponentially small probability.
\begin{lemma}
	\label{lem:upper-bound-machine-i}
	Consider any instance $t$.
	For every machine $i \in N$ and for every $q \in \mathbb R_{>0}$, it holds
	\[
		\mathbb P \left( \theta_i(t) \ge \mu(t) + q \cdot \Opt{t} \right) \le e^{ - q^2 / (2 n) } \enspace .
	\]
\end{lemma}
\begin{proof}
	Consider a fixed machine $i \in N$ for a fixed instance $t$.
	For every task $j \in M$, let 
	\[
		X_j := 
		\begin{cases} 
			t_{ij} 	&,~\text{if}~j \in R_i(t),\\
			0		&,~\text{else}.
		\end{cases}
	\]
	By definition of \textsc{Exp-Bounded-Square} (Algorithm~\ref{alg:exp-bounded-square}), the random variables $X_1, \dots, X_m$ are independent, and it holds $0 \le X_j \le 2 \cdot t_j^{\min}$ for all $j \in M$, where $t_j^{\min}$ is the minimum processing time that any machine has for task $j$.
	Then it holds $\theta_i(t) = \sum_{j \in M} X_j$, and by assumption, $\Opt{t} > 0$. For every $q \in \mathbb R_{>0}$, Hoeffding's inequality~\cite{Hoeffding1963} yields
	\begin{equation}
		\label{ineq:hoeffding}
		\mathbb P \Big( \theta_i(t) - \mathbb E \left[ \theta_i(t) \right] \ge q \cdot \Opt{t} \Big) 
		\le \exp \left( - \frac {2 (q \cdot \Opt{t})^2} {\sum_{j \in M} (2 t_j^{\min})^2} \right)
		=  \exp \left( - \frac {q^2 \cdot (\Opt{t})^2} {2 \cdot \sum_{j \in M} (t_j^{\min})^2} \right) \enspace.
	\end{equation}
	Clearly, for any $j \in M$, it holds $t_j^{\min} \le \max_{j' \in M} t_{j'}^{\min}$.
	Thus,
	\begin{equation}
		\label{ineq:t_j_min_square_sum}
		\sum_{j \in M} (t_j^{\min})^2 
		\le \left(\max_{j \in M} t_j^{\min}\right) \cdot \sum_{j \in M} t_j^{\min} 
		\le \Opt{t} \cdot n \cdot \Opt{t} \enspace,
	\end{equation}
	where we used~\Cref{observation:opt-bounds}.
	Since $\mathbb E[\theta_i(t)] \le \mu(t)$ holds for every machine $i \in N$ due to \Cref{lem:lu_yu_max_expected_load}, we obtain
	\begin{align*}
		\mathbb P \left( \theta_i(t) \ge \mu(t) + q \cdot \Opt{t} \right) 
		&\le 
		\mathbb P \Big( \theta_i(t) \ge \mathbb E[\theta_i(t)] + q \cdot \Opt{t} \Big)\\
		&= 
		\mathbb P \Big( \theta_i(t) - \mathbb E[\theta_i(t)] \ge q \cdot \Opt{t} \Big)\\
		&\overset{(\ref{ineq:hoeffding})}\le
		\exp \left( - \frac {q^2 \cdot (\Opt{t})^2} {2 \sum_{j \in M} (t_j^{\min})^2} \right)\\
		&\overset{(\ref{ineq:t_j_min_square_sum})}\le
		\exp \left( - \frac {q^2 \cdot (\Opt{t})^2} {2 n \cdot (\Opt{t})^2} \right)\\
		&= e^{ - q^2 / (2 n) } \enspace.
	\end{align*}
\end{proof}

By applying the union bound we obtain the following bound on the expected makespan of \textsc{Exp-Bounded-Square}. 
\begin{theorem}
	\label{thm:upper-bound}
	For every instance $t$ and every $q \in \mathbb R_{> 0}$, \textsc{Exp-Bounded-Square} computes a random allocation with an expected makespan of
	\[
		\mathbb E[\Mech{\textsc{EBS}}{t}] \le \Opt{t} \cdot \left( \frac {n+1} 2 + q + 2 n^2 \cdot e^{ - {q^2} / {2 n} } \right) \enspace,
	\]
	where $\Opt{t}$ is the optimal makespan of instance $t$.
\end{theorem}
\begin{proof}
	By a union bound, \Cref{lem:upper-bound-machine-i} implies for every $q \in \mathbb R_{>0}$ that
	\begin{equation}
		\label{ineq:union-bound-makespan}
		\mathbb P \left( \max_{i \in N} \theta_i(t) \ge \mu(t) + q \cdot \Opt{t} \right) \le n \cdot e^{ - q^2 / (2 n) } \enspace.
	\end{equation}
	By definition of \textsc{Exp-Bounded-Square}, it follows
	\[
		\theta_i(t) = \sum_{j \in R_i(t)} t_{ij} \le  \sum_{j \in R_i(t)} 2 t_j^{\min} \enspace.
	\]
	\Cref{observation:opt-bounds} yields a trivial upper bound on the random makespan of \textsc{Exp-Bounded-Square}:
	\[
		\Mech{\textsc{EBS}}{t} = \max_{i \in N} \theta_i(t) \le
		\sum_{j \in M} 2 t_j^{\min}
		\le 2 n \cdot \Opt{t} \enspace.
	\]
	We use this trivial upper bound for the event $\Mech{\textsc{EBS}}{t} \ge \mu(t) +  q \cdot \Opt{t}$.
	Recall that we bounded the probability of that event already in~\Cref{ineq:union-bound-makespan}.
	Therefore,
	\begin{align*}
		\mathbb E[\Mech{\textsc{EBS}}{t}] 
		\le &~ \mathbb P\left (\Mech{\textsc{EBS}}{t} < \mu(t) + q \cdot \Opt{t} \right) \cdot \left(\mu(t) + q \cdot \Opt{t} \right)\\
		& + \mathbb P\left (\Mech{\textsc{EBS}}{t} \ge \mu(t) + q \cdot \Opt{t} \right) \cdot 2n \cdot \Opt{t}\\
		\le &~ \mu(t) + q \cdot \Opt{t} + \mathbb P\left (\Mech{\textsc{EBS}}{t} \ge \mu(t) + q \cdot \Opt{t} \right) \cdot 2n \cdot \Opt{t}\\
		\le &~ \mu(t) + q \cdot \Opt{t} + n \cdot e^{ - q^2 / (2 n) } \cdot 2n \cdot \Opt{t}\\
		= &~  \Opt{t} \cdot \left(  \frac {\mu(t)} {\Opt{t}} + q + 2 n^2 \cdot e^{ - q^2 / (2 n) } \right) \enspace.
	\end{align*}
	Since $\mu(t) = \frac {n+1} 2 \cdot \Opt{t}$, the theorem follows.
\end{proof}

By rewriting $2n^2 \cdot e^{ - q^2 / (2 n) } = 2e^{2 \ln n - q^2 / (2n)}$, we finally obtain the claimed upper bound on the (expected) approximation ratio as a corollary, by selecting $q = 2 \sqrt{n \ln n}$.

\begin{corollary}\label{cor:upper-bound}
	For every instance $t$, it holds
	\[
		\mathbb E[\Mech{\textsc{EBS}}{t}] \le \Opt{t} \cdot \left( \frac {n+1} 2 + 2 \sqrt{n \ln n} + 2 \right) \enspace.
	\]
\end{corollary}

	\bibliography{main}
	
	\appendix
	\newcommand{\slicedoff}{chopped off}
\newcommand{\Slicedoff}{Chopped off}
\newcommand{\slicedoffx}{strictly chopped off}
\newcommand{\Description}[1]{\vspace{1em}#1}
\newcommand{\opt}{\ensuremath{\textsc{Opt}}}
\newcommand{\mech}{\ensuremath{\textsc{Mech}}}

\section{Adaptation of the Box Theorem of \cite{ChristodoulouKK26}}

\label{sec:box}

This section contains the precise statement and adapted proof of a generalized version of the \emph{Box Theorem} that appeared in \cite{ChristodoulouKK26}. The original Box Theorem has been used for lower bounding the approximability of unrelated scheduling by deterministic truthful mechanisms. The original proof exploits at several points that the arbitrary hypothetical deterministic weakly monotone allocation algorithm at hand has approximation ratio lower than $n$ (otherwise the lower bound proof is done). Since our result considers randomized mechanisms, we need to prove the Box Theorem for \emph{deterministic}  weakly monotone allocations having a \emph{fixed} arbitrarily high upper bound $K$ on the approximation ratio (cf. Section~\ref{sec:prelim}).
 
Although the necessary modifications in the proof are straightforward, given the length of the proof, we include here the adaptation of the Box Theorem for completeness. We provide here most parts (e.g., notation, intuition) of the proof from  \cite{ChristodoulouKK26} that we found necessary for consistency and readability independently of the original paper. Many parts of the original proof (e.g., the base case of the induction)  that exploit only \emph{bounded} approximation (vs. $n$-approximation) hold without change, and we omit these parts. Roughly, we recite those lemmas and proofs where \emph{some} occurrences of the parameter $n$ have to be replaced by the parameter $K.$ Apart from omissions, we change the original text as little as possible, for comparability with \cite{ChristodoulouKK26}.

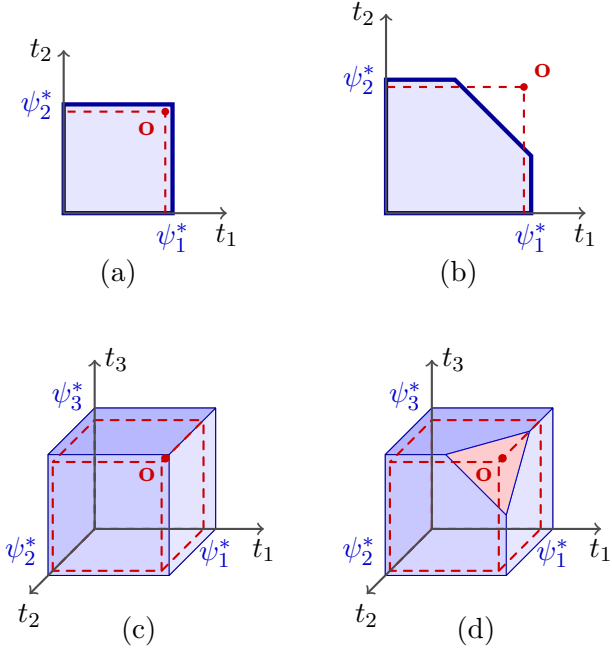
\begin{figure}[h]
\newcommand{\Depth}{2}
\newcommand{\Height}{2}
\newcommand{\Width}{2}
\newcommand{\Fraction}{0.5}
\newcommand{\deltaX}{0.2}
\newcommand{\deltaY}{0.2}
\newcommand{\psiStar}{\psi^*}

  \centering

  \begin{tikzpicture}[scale=0.48]
    \draw[draw=blue!60!black, fill=blue!10, ultra thick] (0,0) -- (3,0) -- (3,3) -- (0,3) -- cycle;
    
    \draw[->, thick, black!70] (0,0) -- (4.5,0) node[anchor=north, text=black] {$t_{1}$};
    \draw[->, thick, black!70] (0,0) -- (0,4.5) node[anchor=east, text=black] {$t_{2}$};

    \draw[thick, dashed, red!80!black] (3-\deltaY,0)  -- (3-\deltaY,3-\deltaY) -- (0,3-\deltaY) ;
    \fill[red!80!black] (3-\deltaY,3-\deltaY) circle (3pt) node[anchor=north east] {$\mathbf{o}$};

    \draw[blue!80!black] (3,0) node[anchor=north] {$\psiStar_1$};
    \draw[blue!80!black] (0,3) node[anchor=east] {$\psiStar_2$};
    \draw (1.5,-1) node[anchor=north] {(a)};
  \end{tikzpicture}\hspace{1cm}
  \begin{tikzpicture}[scale=0.48]
    \draw[draw=blue!60!black, fill=blue!10, ultra thick] (0,0) -- (4,0) -- (4,1.58) -- (1.9,3.68) -- (0,3.68) -- cycle;
    
    \draw[->, thick, black!70] (0,0) -- (5.5,0) node[anchor=north, text=black] {$t_{1}$};
    \draw[->, thick, black!70] (0,0) -- (0,5.5) node[anchor=east, text=black] {$t_{2}$};

    \draw[thick, dashed, red!80!black] (4-\deltaY,0)  -- (4-\deltaY,3.68-\deltaY) -- (0,3.68-\deltaY) ;
    \fill[red!80!black] (4-\deltaY,3.68-\deltaY) circle (3pt) node[anchor=south west] {$\mathbf{o}$};

    \draw[blue!80!black] (4,0) node[anchor=north] {$\psiStar_1$};
    \draw[blue!80!black] (0,3.68) node[anchor=east] {$\psiStar_2$};
    \draw (2,-1) node[anchor=north] {(b)};
  \end{tikzpicture}

  \vspace{0.5cm}

  \begin{tikzpicture}[scale=0.8, line join=round, line cap=round]
    \coordinate (O) at (0,0,0);
    \coordinate (A) at (0,\Width,0);
    \coordinate (B) at (0,\Width,\Height);
    \coordinate (C) at (0,0,\Height);
    \coordinate (D) at (\Depth,0,0);
    \coordinate (E) at (\Depth,\Width,0);
    \coordinate (F) at (\Depth,\Width,\Height);
    \coordinate (G) at (\Depth,0,\Height);

    \coordinate (O1) at (0,0,0);
    \coordinate (A1) at (0,\Width-\deltaX,0);
    \coordinate (B1) at (0,\Width-\deltaX,\Height-\deltaX);
    \coordinate (C1) at (0,0,\Height-\deltaX);
    \coordinate (D1) at (\Depth-\deltaX,0,0);
    \coordinate (E1) at (\Depth-\deltaX,\Width-\deltaX,0);
    \coordinate (F1) at (\Depth-\deltaX,\Width-\deltaX,\Height-\deltaX);
    \coordinate (G1) at (\Depth-\deltaX,0,\Height-\deltaX);

    \draw[draw=blue!60!black, fill=blue!5] (O) -- (C) -- (G) -- (D) -- cycle;
    \draw[draw=blue!60!black, fill=blue!15] (O) -- (A) -- (E) -- (D) -- cycle;
    \draw[draw=blue!60!black, fill=blue!25] (O) -- (A) -- (B) -- (C) -- cycle;
    \draw[draw=blue!60!black, fill=blue!10, opacity=0.85] (D) -- (E) -- (F) -- (G) -- cycle;
    \draw[draw=blue!60!black, fill=blue!20, opacity=0.75] (C) -- (B) -- (F) -- (G) -- cycle;
    \draw[draw=blue!60!black, fill=blue!30, opacity=0.85] (A) -- (B) -- (F) -- (E) -- cycle;

    \draw[thick, dashed, red!80!black] (O1) -- (D1) -- (G1) -- (C1) -- cycle; %
    \draw[thick, dashed, red!80!black] (A1) -- (E1) -- (F1) -- (B1) -- cycle; %
    \draw[thick, dashed, red!80!black] (O1) -- (A1) (D1) -- (E1) (G1) -- (F1) (C1) -- (B1); %

    \coordinate (AA) at (0,1.4*\Width,0);
    \coordinate (CC) at (0,0,1.4*\Height);
    \coordinate (DD) at (1.4*\Depth,0,0);

    \draw[->, thick, black!70] (O) -- (AA);
    \draw[->, thick, black!70] (O) -- (CC);
    \draw[->, thick, black!70] (O) -- (DD);

    \draw (2.8,0,0) node[anchor=north] {$t_1$};
    \draw (0,2.8,0) node[anchor=west] {$t_3$};
    \draw (0,0,2.8) node[anchor=north] {$t_2$};

    \draw[blue!80!black] (2,0,0) node[anchor=north] {$\psiStar_1$};
    \draw[blue!80!black] (0,2.2,0) node[anchor=east] {$\psiStar_3$};
    \draw[blue!80!black] (0,0,2) node[anchor=south east] {$\psiStar_2$};

    \fill[red!80!black] (1.9,1.9,1.9) circle (2pt) node[anchor=north east] {$\mathbf{o}$};

    \draw (1.5,-0.5, 2) node[anchor=north] {(c)};
  \end{tikzpicture}\hspace{0.5cm}
  \begin{tikzpicture}[scale=0.8, line join=round, line cap=round]
    \coordinate (O) at (0,0,0);
    \coordinate (A) at (0,\Width,0);
    \coordinate (B) at (0,\Width,\Height);
    \coordinate (C) at (0,0,\Height);
    \coordinate (D) at (\Depth,0,0);
    \coordinate (E) at (\Depth,\Width,0);
    \coordinate (F) at (\Depth,\Width,\Height);
    \coordinate (G) at (\Depth,0,\Height);

    \coordinate (O1) at (0,0,0);
    \coordinate (A1) at (0,\Width-\deltaX,0);
    \coordinate (B1) at (0,\Width-\deltaX,\Height-\deltaX);
    \coordinate (C1) at (0,0,\Height-\deltaX);
    \coordinate (D1) at (\Depth-\deltaX,0,0);
    \coordinate (E1) at (\Depth-\deltaX,\Width-\deltaX,0);
    \coordinate (F1) at (\Depth-\deltaX,\Width-\deltaX,\Height-\deltaX);
    \coordinate (G1) at (\Depth-\deltaX,0,\Height-\deltaX);

    \coordinate (BF) at (\Fraction*\Depth,\Width,\Height);
    \coordinate (EF) at (\Depth,\Width,\Fraction*\Height);
    \coordinate (GF) at (\Depth,\Fraction*\Width,\Height);

    \draw[draw=blue!60!black, fill=blue!5] (O) -- (C) -- (G) -- (D) -- cycle;
    \draw[draw=blue!60!black, fill=blue!15] (O) -- (A) -- (E) -- (D) -- cycle;
    \draw[draw=blue!60!black, fill=blue!25] (O) -- (A) -- (B) -- (C) -- cycle;
    \draw[draw=blue!60!black, fill=blue!10, opacity=0.85] (D) -- (E) -- (EF) -- (GF) -- (G) -- cycle;
    \draw[draw=blue!60!black, fill=blue!20, opacity=0.75] (C) -- (B) -- (BF) -- (GF) -- (G) -- cycle;
    \draw[draw=blue!60!black, fill=blue!30, opacity=0.85] (A) -- (B) -- (BF) -- (EF) -- (E) -- cycle;
    
    \draw[draw=blue!80!black, fill=red!20, opacity=0.95] (BF) -- (EF) -- (GF) -- cycle;

    \draw[thick, dashed, red!80!black] (O1) -- (D1) -- (G1) -- (C1) -- cycle; %
    \draw[thick, dashed, red!80!black] (A1) -- (E1) -- (F1) -- (B1) -- cycle; %
    \draw[thick, dashed, red!80!black] (O1) -- (A1) (D1) -- (E1) (G1) -- (F1) (C1) -- (B1); %

    \coordinate (AA) at (0,1.4*\Width,0);
    \coordinate (CC) at (0,0,1.4*\Height);
    \coordinate (DD) at (1.4*\Depth,0,0);

    \draw[->, thick, black!70] (O) -- (AA);
    \draw[->, thick, black!70] (O) -- (CC);
    \draw[->, thick, black!70] (O) -- (DD);

    \draw (2.8,0,0) node[anchor=north] {$t_1$};
    \draw (0,2.8,0) node[anchor=west] {$t_3$};
    \draw (0,0,2.8) node[anchor=north] {$t_2$};

    \draw[blue!80!black] (2,0,0) node[anchor=north] {$\psiStar_1$};
    \draw[blue!80!black] (0,2.2,0) node[anchor=east] {$\psiStar_3$};
    \draw[blue!80!black] (0,0,2) node[anchor=south east] {$\psiStar_2$};

    \fill[red!80!black] (1.9,1.9,1.9) circle (2pt) node[anchor=north east] {$\mathbf{o}$};

    \draw (1.5,-0.5, 2) node[anchor=north] {(d)};

  \end{tikzpicture}

  \caption{\small \textbf{Box Definition.} Allocation regions where the root gets all tasks in a star network with $2$ leaves (a--b) and $3$ leaves (c--d). For root $i$ and leaves $j$, $t_j$ denotes the root's values, and $\psi_j^*=\psi_{j}(s_j)$ represents the boundary values. The dashed red lines correspond to the thresholds $\psiStar_j-\delta$. Cases (a) and (c) are valid boxes because the corner $\mathbf{o}$ lies inside the target allocation region. Conversely, cases (b) and (d) are not valid boxes since $\mathbf{o}$ falls outside this region.}
  \label{fig:box}
\end{figure}

The Box Theorem~\ref{thm:box-restated} states
that in all multi-stars with sufficiently high multiplicity  the probability that a random star is a \emph{box} tends to 1, as the multiplicity tends to infinity. Roughly, given a deterministic mechanism, and a trivial\footnote{that is, all tasks have a $0$ entry} multi-graph instance, a \emph{box} is a star $S$ for which
when we fix the values of its leaves, %
the allocation region $R_S$ of the
root for getting all the tasks in $S$ is rectangular (see
Figure~\ref{fig:box} for an illustration, Definition~\ref{def:box} for
the precise definition, and Section~\ref{sec:region-r_p} for a precise
definition and properties of the allocation region $R_S$).
We start this section by introducing the basic notation, formulating the box theorem and giving an outline of  the proof.
Most of the following definitions depend on the mechanism at hand, therefore it will be convenient to \emph{fix an arbitrary deterministic truthful mechanism  throughout the section}. Let $K\geq n+1$ be an arbitrary fixed upper bound on the approximation ratio of the mechanism.  %

\subsection{Statement of the Box Theorem}
 
We consider multi-stars with $n$ nodes. Since we are dealing with multi-stars and not general multi-graphs, it will be convenient to simplify the notation: 
the root is node 0 and the leaves are nodes $1,\ldots, n-1$. We name the edges $1,2,\ldots, m$, where $m=\ell(n-1)$. Edges correspond to tasks; two parallel edges are called \emph{siblings}. The set of tasks for leaf $i$ is denoted by $C_i$, and we use the term \emph{leaf} for both $i$ and $C_i$.  The multiplicity of every edge is $\ell$, i.e., $|C_i|=\ell$, for every $i\in [n-1]$.

An edge $j\in [m]$
between node 0 and leaf $i$ has two nonnegative values, %
which represent the processing time for the two players, $t_j$ for the root and $s_j$ for the leaf player. 
The set of values for all edges $T=(t,s)=(t_j, s_j)_{j\in [m]}$ is called an instance.
All the instances in this section satisfy $s_j\in [0,B)$ for some arbitrarily high value $B$, although in most of the argument $s_j$ takes values in $[0,1]$.\footnote{We need to fix $B$ and $K$ among others, to make sure that the mechanism allocates the edges to the dedicated players, corresponding to the two incident vertices, for example by setting the values of other players to at least $K\cdot B$ for these tasks.} In particular, the Box Theorem is about multi-stars with values  $s_j\in(\xi,1)$, where $\xi$ can be any strictly positive value, in particular we can set  $\xi=1-\varepsilon$ according to the notation of Section~\ref{sec:lb}. Naturally the argument applies to any interval of values not only to $(\xi, 1).$

For a given instance $T=(t,s)$, we denote the boundary function
of the root for edge $j\in [m]$ by
$\psi_j(s_j)$. Recall that the interpretation is that, having fixed the values of
all other edges, \emph{edge $j$ is given to the root if $t_j<
  \psi_j(s_j)$, and it is given to the leaf when
  $t_j>\psi_j(s_j)$}. Since the argument sometimes considers more than
one instance, it will be useful to extend the notation to
$\psi_j(t_{-j}, s_{-j}, s_j)$ to explicitly indicate the values of the
other edges. We also write it as $\psi_j[t_{-j},s_{-j}](s_j)$ when we
want to treat it as a function of $s_j$ only. Since the values of most
other edges can be inferred from the context, we only indicate the
values that have changed inside the optional part (the part inside the
square brackets). For example, for an instance $T=(t, s)$, which can
be inferred from the context, if we change $t_1$ from its current
value to $x$, the new boundary function of edge 2 is written as
$\psi_2[t_1=x](s_2)$ or simply $\psi_2[x](s_2)$.

Recall that the instances of the argument are selected randomly so that they satisfy the \emph{continuity requirement} almost surely (although their values do not have to be rational). %
We use the following fixed parameters.

\begin{definition}[Parameters] \label{def:parameters} The values of these parameters are assumed to be rational numbers.
  \begin{itemize}
  \item $\xi\in (0,1);$
  \item $\nu\in (0,\frac{\xi}{nK\cdot 4^n})$, a very small fixed strictly positive value; %
  \item $\nu'=\nu/4$.
  \end{itemize}
\end{definition}

\begin{definition}[Box]\label{def:box}
  For a fixed $\nu$, and instance $(t,s),$ a star $S$ with $k$ leaves is called a \emph{$\delta$-box or simply box}, if the mechanism allocates all edges to the root, when we set $t_j=\psi_j(s_j)-\delta$ for every leaf $j$ of $S$, where $\delta=4^{k}\nu$.
\end{definition}

The aim is to prove that a multi-star $(t,s)$ with $t_j=0$ and $s_j\in (\xi, 1)$, for all $j\in[m]$ has many boxes (in fact, almost all stars are boxes) with parameter $\delta=4^n\nu$. To apply induction, we need the definition of box for smaller stars.  Note that $\delta$ depends on the size of the star.

The following related instances are used many times in what follows. %

\begin{definition}[Instance $T_\nu(S)$]\label{def:nu}
  For a given instance $T=(t,s)$, let $T_\nu(S)=(t^\nu, s)$ be the instance that agrees with $T$ everywhere, except for tasks in a star $S$ with $k$ leaves for which $t_j = \psi_j(s_j)-4^k\nu$.  
\end{definition}

We can now state the Box Theorem.

\begin{theorem}[Box] \label{thm:box-restated} Fix a mechanism with approximation ratio less than $K.$ For every $\nu$ and $\xi$
 that satisfy the requirements of
  Definition~\ref{def:parameters}, in a multi-star instance $T=(t,\bar
  s)$ with values $t_j=0$, $\bar s_j\in [\xi,1],$ and multiplicity $\ell$ for each leaf,  that satisfies the continuity requirement, 
   the probability that a random star of $n-1$ leaves is not a $\delta$-box, is at most $\left(\frac{5K}{\nu}\right)^{n-1}\frac{1}{\sqrt{\ell}},$
  where $\delta=4^{n-1}\nu$.

\end{theorem}

\subsubsection{Intuition.} 
The proof of the Box Theorem is by induction on the number of leaves $k$. The case of one leaf is trivial --- all stars are boxes. However, we don't use this as the base case of the induction, because proving the theorem for $k=2$ requires different handling than the general case.%

Roughly speaking the inductive step goes as follows: suppose that we have established that most stars of multi-stars with $k-1$ leaves are boxes. Then we can find a star $S$ of $k$ leaves --- actually many such stars --- in which all its sub-stars of $k-1$ leaves are boxes. Now $S$ may be a box itself,  or a box with a single corner cut off by a diagonal cut (see Figure~\ref{fig:box}). Let's call this shape ``\slicedoff\ box''.

The core of the argument relies on the characterization of $2\times 2$ (2 players, 2 tasks) mechanisms (see Section~\ref{sec:facts-about-two}). This characterization is crucial for establishing that either the \slicedoff\ box is actually a box, or that we can obtain a box by replacing one of the tasks in $S$ with a sibling task (i.e., a task from the same leaf). To show this, we select a task $p$ from the star $S$ and a sibling $p'$ of $p$, and consider the $(p,p')$-slice mechanism, that is, the mechanism for these two tasks when the values of all other tasks are fixed. By the characterization of $2\times 2$ mechanisms, the slice mechanism is either a (relaxed) affine minimizer or a (relaxed) task-independent mechanism.

If the slice mechanism is task-independent the proof is relatively straightforward. The fact that tasks $p$ and $p'$ are independent implies by geometry of the allocation region that if we replace task $p$ by $p',$ %
we obtain a \emph{box} that includes $p'$ and the remaining $k-1$ tasks. 
The only real complication for this case would arise at discontinuity points, but we have excluded them by the continuity requirement.

The other case is when the $(p,p')$-slice mechanism is an affine minimizer and the allocation boundaries are linear functions. The key idea is to exploit this linearity. We keep only the linearity property for $p$ and completely ignore task $p'$ and focus on the original box. Linearity is used to show that the allocation boundaries move rectilinearly, when we change the leaf value of task $p$. In other words, the whole upper envelope of the \slicedoff\ box moves rectilinearly. But then if it is moved sufficiently close to the side, the \slicedoff\ piece will reach the boundary and create a \slicedoff\ box on it (see Figure 5 in \cite{ChristodoulouKK26}). This would contradict the inductive hypothesis that the sub-stars are boxes or equivalently that the sides are lower dimensional boxes, except for the fact that we have changed the leaf value of task $p$. To actually reach a contradiction, we use a stronger inductive hypothesis in which the side (sub-star) is a box for many values of the leaf (see the definition of \slicedoff\ box, below).

Notice however that the above argument fails for $k=2$, because a \slicedoff\ 1-dimensional interval is still an interval, that is, a box. Thus the case of $k=2$ needs to be handled differently. For this, we consider a star with two leaves and two edges per leaf and show that at least one of its four stars is a box. This is done by an argument similar to the argument for the inductive step, only that linearity is exploited for moving boundaries in two distinct directions. The fact that a multi-star with two leaves and two edges per leaf contains a box, can be used %
 to show that almost all stars are boxes.

The rest of this section makes the above outline precise. In particular, in Subsection~\ref{sec:def-notation} we  provide further definitions,  and useful facts and lemmas. In Section~\ref{sec:box_main_argument} we present the main argument. In Section~\ref{sec:inductionbase} we provide the proof of the base case, in Section~\ref{sec:induction_step} the proof of the induction step, and finally in Section~\ref{sub:existence} we wrap up to show the existence of a box.

\subsection{Further notation and preliminaries.} 
\label{sec:def-notation}

Usually we fix the values of most tasks and consider the allocation of the mechanism on the remaining tasks. In particular, we do this when we employ the characterization of $2\times 2$ mechanisms. The next definition formalizes this.

\begin{definition}[Slice and slice mechanism] \label{def:slice} Fix an instance $T$ and two tasks $p$ and $p'$. The set of instances that agree with $T$ on all tasks except for the tasks $p$ and $p'$ is called a \emph{$(p,p')$ slice}. The allocation function of these two tasks by the mechanism is called the $(p,p')$ \emph{slice mechanism} for the given values of the other tasks. %
  Similarly, we define slice mechanisms for larger sets of tasks.
\end{definition}

\begin{definition}[Trivial leaf]
  A leaf $i$ with set of edges $C_i$ is called \emph{trivial} for a given instance $T$ if $t_j=0$ for every task $j\in C_i.$ We also call a task $j$ \emph{trivial} if $t_j=0.$
\end{definition}

We now provide the definition of the main type of instances that we consider throughout.

\begin{definition}[Standard instance] \label{def:standard} %
  For a given truthful mechanism, an instance $T=(t,\bar s)$ is a \emph{standard instance for a set of leaves} $\cal C$ if the following conditions hold
  \begin{itemize}
  \item $t_j=0$, for every $j\in [m]$
  \item $\bar s_j\in (\xi,1)$, for every task $j$ of these leaves, i.e., $j\in \cup_{i\in \cal C}C_i$,
  \item $T$ satisfies the continuity requirement. 
  \end{itemize}
  The leaves in $\cal C$ are then called \emph{standard leaves} for the instance $T.$ 
\end{definition}

Henceforth the notation $\bar s$ will denote some fixed standard instance (clear from the context). We reserve the notation $s$ for the variable, which can take any value. Note that instance $T=(t,\bar s)$ in the statement of Theorem~\ref{thm:box-restated} is standard for the set of all leaves.

The central part of the argument is an induction on the number of
leaves $k$. In the induction step from $(k-1)$ to $k$, the remaining
leaves are trivial with fixed values, and therefore play no role; in
particular they do not affect the approximation ratio. %

\begin{definition}[critical values $\alpha_j$ for singletons]\label{def:alphai}
  For standard instances $T=(t, \bar s)$, we will use the shorthand $\alpha_j$ for $\psi_j(\bar s_j),$ if $\bar s$ is fixed and clear from context.
\end{definition}

\begin{lemma} \label{obs:alphas} Let $T=(t,s)$ be an instance so that all leaves are trivial. Assuming that the approximation factor is less than $K$,  for every task $j\in [m]$  it holds that 
\begin{enumerate}
\item[(i)] $s_j/K<\psi_j(s_j)<K s_j;$%
\item[(ii)] $\lim_{s_j\rightarrow 0}\psi_j(s_j)=\psi_j(0)=0;$ 
\end{enumerate}
\end{lemma}

\begin{proof}

 If for some $x$, $\psi_{j}(x)\geq K\cdot x$, by setting the values of task $j$ for nodes $i$ and $0$ to $x$ and $\psi_{j}(x)-\eta x$, respectively, for some $\eta>0$, the approximation ratio is at least $K-\eta$. The proposition follows by letting $\eta\rightarrow 0$. The proof of the lower bound is analogous. 
 Finally, the statement (ii) follows from (i).
\end{proof}

\subsubsection{Region $R_P$}
\label{sec:region-r_p}

The next definition of region %
$R_P$ for a given mechanism and given instance $T$ concerns a $k$-dimensional slice defined by a task set $P.\,$ The other tasks $[m]\setminus P$ have fixed values and for simplicity they are not treated in the definition. In most cases when $R_P$ is considered, $P$ will be a star, or a pair of siblings, and  the tasks $[m]\setminus P$ will be trivial.

\begin{definition}[region $R_P$ and $R_{\emptyset|P}$]\label{def:regionR} Let $T=(t,s)$ be a given instance and $P=\{p_1,\ldots,p_k\}$ be a set of tasks. The allocation region  $R^T_P\subset \mathbb R^k_{\geq 0}$ consists of all vectors $t'_P$ such that for input $T'=((t'_P, t_{-P}), s),$ all tasks of $P$ are allocated to the root. Similarly, $R^T_{\emptyset|P}$ consists of all the $t'_P$ so that for $((t'_P, t_{-P}), s),$ all tasks of $P$ are allocated to the leaves. $T$ is omitted from the notation when it is clear from the context.
\end{definition}

It follows directly by the Weak Monotonicity property that $R_P $ is a
(possibly degenerate) $k$-dimensional polyhedron defined by linear
constraints of the form $\sum_{i\in I} t_{p_i}\leq c_I,$ for some
$I\subset [k]$, and some $c_I\in \mathbb R_+$. In particular, for
$k=2$, $R_P$ is either a rectangle or a rectangle from which we cut
off a piece by a $-45^o$ cut. For higher dimensions, it is a
hyperrectangle (orthotope) with pieces cut off by specific
hypercuts. See Figure~\ref{fig:box} for examples of $R_P$ for $k=2$
and $k=3$ and see also~\cite{Vidali2009} for a geometric interpretation of
truthfulness. Note that $R_P\subseteq \times_{i=1}^k[0,\alpha_{p_i}]$,
where $\alpha_{p_i}=\psi_{p_i}(s_{p_i})$.

Two particular shapes of $R_P$ play significant role in the argument,
boxes (complete orthotopes), and \slicedoff\ boxes, which are boxes
with a single corner removed by a diagonal cut of the form $\sum_{i
  \in [k]} t_{p_i} = c_{[k]}$. Strictly speaking, in the definition of
box and \slicedoff\ box, we allow thin pieces of width at most
$\delta$ to be missing from its faces.

More generally, when the facet $\sum_{i\in [k]} t_{p_i}=c_{[k]}$ of
$R_P$ exists, we call it \emph{bundling facet}. When the bundling
boundary facet exists, it separates the regions $R_P$ and
$R_{\emptyset|P}$.

The following two lemmas will be useful later.

\begin{lemma}\label{obs:Rnonempty}(\cite{ChristodoulouKK26} Lemma 32.) Let $T=(t,s)$ be an instance with
  $t_j=0$ for all $j$, and assume that for some set of tasks
  $P=\{p_1,p_2,\ldots, p_k\}$ the region $R_P$ has a full dimensional
  bundling facet that, in particular, contains a point with strictly positive
  coordinates. Then $\psi_{p_j}(s_{p_j})>0$ and $s_{p_j}>0$ for every
  task $p_j\in P$.
\end{lemma}

\begin{lemma}\label{prop:lipschitz}(\cite{ChristodoulouKK26} Lemma 33.)
  For every truthful mechanism, and arbitrary $T=(t,s),$ the boundary
  function $\psi_r(t_{-r},s)$ is 1-Lipschitz in $t_{-r}$, i.e.,
  $|\psi_r(t_{-r},s)-\psi_r(t_{-r}',s)|\leq |t_{-r}-t_{-r}'|_1$.
  \end{lemma}

\subsubsection{Facts about two tasks.} \label{sec:facts-about-two}

In this section we summarize known concepts and results for two tasks, i.e., two edges that can belong to the same leaf (sibling tasks) or to two different leaves (star of two tasks). Here we assume that all other tasks are fixed, and omit them from the notation. So let simply $t=(t_1, t_2)$ and $s=(s_1,s_2)$.
First we consider allocation regions for the root depending on his own
bids $(t_1,t_2)$ for fixed values $s=(s_1,s_2)$. Here the $s_1$ and $s_2$ do
not necessarily belong to the same leaf. It is known that in the case
of two tasks, for a fixed $s$ the four possible allocation regions of
the root in a weakly monotone allocation subdivide $\mathbb R_{\geq
  0}^2$ basically in three possible ways summarized in the following definition (see
Figure~6 in \cite{ChristodoulouKK26}).

\begin{definition}\label{def:crossing}
  For given $(s_1,s_2)$ we call the allocation for the root 
  \begin{itemize}
  \item \emph{quasi-bundling}, if there are at least two points $t\neq t'$ on
    the boundary of $R_{\{1,2\}}$ and $R_{\emptyset|\{1,2\}};$ 
  \item \emph{quasi-flipping}, if $R_{\{1,2\}}$ and $R_{\emptyset}$ have no common boundary point;  
  \item \emph{crossing}, otherwise (that is, if there is a unique common boundary point).
  \end{itemize}
  We sometimes refer collectively to both quasi-bundling and quasi-flipping allocations as \emph{non-crossing}.
  \end{definition}

  In case the allocation is quasi-bundling, the boundary between $R_{\{1,2\}}$ and $R_{\emptyset|\{1,2\}}$ is a bundling boundary.  In general the facets defining $R_{\{1,2\}}$ in $\mathbb R^2_{\geq 0}$ are of the possible forms $t_1=c_1,\, t_{2}=c_2,\,$ and $t_1+t_{2}= c.$ A nontrivial facet $t_1+t_{2}= c$ exists if and only if the allocation is quasi-bundling.

  We need to distinguish and treat separately a special sort of
  quasi-bundling allocation, because in this case the boundary
  function $\psi_1$ can become nonlinear function of $s_1$, even in affine
  minimizers  (more precisely, in their generalizations called \emph{relaxed} affine minimizers, see Theorem~\ref{theo:addchar} and \cite{ChristodoulouKoutsoupiasKovacs2020Submodular}).
  If the facet $t_1+t_2= c$ exists but the facet $t_1=c_1$ is missing, then we will call the allocation of the tasks
  half-bundling, as defined next (see
Figure~6.d in \cite{ChristodoulouKK26}).

  \begin{definition}\label{def:half-bundling} The allocation of two tasks $(1,2)$ (siblings or not) is called \emph{half-bundling at task $1$,} if the region $R_{\{1,2\}}$ is nonempty and is defined by at most two facets $t_{2}=c_2$ and $t_1+t_{2}=c,$ but no nontrivial facet $t_1=c_1$ exists i.e., $c_1\geq c$.
    The allocation will be called \emph{fully bundling,} if it is half-bundling at both task $1$ and task $2$ i.e. both $c_1,c_2\geq c$.
\end{definition}

Finally, we include here the characterization result for $2$ machines
and $2$ tasks ($2\times 2$ case).   That is, from here on
we consider two sibling tasks: the input $(t_1,t_2)$ is of the root,
and $(s_1,s_2)$ belong to the same leaf. The characterization assumes
that for fixed values $s$ and sufficiently high values of the root,
both tasks will be allocated to the leaf; otherwise the approximation
ratio is unbounded (even when there exist other players and tasks with
fixed values). Recall that in our instances $T=(t,s)$ the values for
the root can be arbitrarily high, but for the $s$ players they belong in
$[0,B)$. %
The theorem from \cite{ChristodoulouKoutsoupiasKovacs2020Submodular} that we use, characterizes mechanisms
with values exactly in this domain. For a description of these
  mechanisms see \cite{ChristodoulouKoutsoupiasKovacs2020Submodular}. For related characterization
  results see \cite{ChristodoulouKoutsoupiasVidali2008, DobzinskiSundararajan2008}. 

\begin{theorem}[Characterization of $2\times 2$ mechanisms~\cite{ChristodoulouKoutsoupiasKovacs2020Submodular}] \label{theo:addchar} For every $B\in \R_{>0}\cup\{\infty\}$, every weakly monotone allocation for two tasks and two players with values $t\in [0,\infty)\times [0,\infty)$ and $s\in[0,B)\times [0,B)$ --- such that for every $s$ there exists $t$ for which both tasks are allocated to the second player --- belongs to the following classes: (1) relaxed affine minimizers (including the special case of affine minimizers), (2) relaxed task-independent mechanisms (including the special case of task-independent mechanisms), (3) 1-dimensional mechanisms, (4) constant mechanisms.
\end{theorem}

We remark that all of these mechanisms are (special cases of) relaxed task-independent or relaxed affine minimizer mechanisms. We summarize the most relevant properties of these  mechanisms in terms of the possible allocation figures (see \cite{ChristodoulouKoutsoupiasKovacs2020Submodular}). Roughly speaking, when there exists a bundling (or flipping) boundary and three or four allocation regions, the boundary functions are affine.

\begin{observation}\label{obs:possibleFigures}
  The following properties hold for a $2\times 2$ truthful mechanism.
\begin{itemize}
\item[(i)] the allocation of a relaxed task-independent mechanism is crossing for every $(s_1,s_2),$ except for countably many points $(s_1',s_2'),$ where both $\psi_1 (s_1)$ has jump discontinuity in $s_1',$ \emph{and} $\psi_2 (s_2)$ has jump discontinuity in $s_2';$

\item[(ii)] the allocation of a (non task-independent) relaxed affine minimizer is non crossing for every $(s_1, s_2);$ %
either it is always quasi-bundling, or always quasi-flipping, with the same length of slanted boundary for every high enough $s$ --- as $s_1$ and/or $s_2$ gets smaller, part of the slanted boundary may disappear, thus it may get shorter (moreover, in degenerate relaxed affine minimizers, the possible length of the slanted boundary may be unbounded); %

\item[(iii)] the boundary function $\psi_1[t_{2},s_{2}](s_1)$ of a relaxed affine minimizer is a truncated linear function for every $(t_2, s_1, s_2),$ unless the allocation is fully bundling: %
\begin{align*}
\psi_1[t_2,s_2](s_1)=\max(0\,, \,\lambda(t_2,s_2)\, s_1 - \gamma(t_2,s_2)\,).
\end{align*}
      
Symmetric statements hold for $\psi_2.$       
\end{itemize}
\label{ref:obs-characterization}
\end{observation}

The fact that relaxed task-independent mechanisms have discontinuities could create complications, but we avoid them by considering instances that satisfy the continuity requirement.

\subsection{Proof of the main Box Theorem~\ref{thm:box-restated}}
\label{sec:box_main_argument}

\begin{proof}[Proof of Theorem~\ref{thm:box-restated}]
  Fix some standard instance of $n-1$ leaves and sufficiently high multiplicity $\ell$. We show by induction on $k$ that in every subset of $k\in [n-1]$ leaves, \emph{a random star selected uniformly and independently is a box} with arbitrarily high probability.
More precisely, let $1-b_k$ be (a lower bound on) the probability that a random star of $k$ leaves is a box. We establish a recurrence on $b_k$ that shows that $b_{n-1}\rightarrow 0$ when $\ell \rightarrow \infty.$ %

Specifically, Theorem~\ref{cor:b2} establishes the base case of the induction ($k=2$), which shows that $b_2\leq 2/\sqrt{\ell}$. Lemma~\ref{lemma:recurrence} based on the proof of the inductive hypothesis establishes the recurrence $$b_k\leq \left (\frac{5K}{\nu}-1 \right )b_{k-1}+ \frac{2nK}{\xi \sqrt{\ell} },$$ for $k\geq 3$. It follows that for sufficiently large $\ell$, $b_{n-1}$ can be arbitrarily small (see~Corollary~\ref{cor:bound-on-b} for a more precise bound on $b_k$).
\end{proof}

\subsubsection{Preliminary observations}
\label{sec:preliminary}

The following theorem %
says that if the mechanism has bounded approximation ratio, and given that all other tasks are trivial, the allocation of the slice mechanism of two siblings must be crossing. %

\begin{theorem} \label{thm:sibling-independence} Let $T=(t,\bar s)$ be
  a standard instance. If the mechanism has approximation ratio at
  most $K$, the boundary function value $\psi_p[t_{-p},\bar s_{-p}](s_p)$ of
  a task $p$ is independent of the running time values of its
  siblings. %
\end{theorem}

\begin{proof}
  The theorem is a direct consequence of the following lemma. %
  The lemma states that slice mechanism of two sibling tasks cannot
  have a quasi-bundling or quasi-flipping boundary, when the
  approximation ratio is bounded. By the characterization, the slice
  mechanism must be a relaxed task-independent mechanism. By the
  continuity requirement, there are no discontinuities, so
  $\psi_p[t_{-p},\bar s_{-p}](s_p)$ of a task $p$ is independent of
  the values of its siblings. 
  
  Since $T$ is a standard instance, the other tasks are trivial. So an unbounded approximation ratio for a $(p,p')$ slice would imply unbounded ratio for the whole mechanism.
\end{proof}

\begin{lemma}\label{prop:quasi} Let $T=(t,\bar s)$ be a standard instance, and $Q$ be any leaf. Assume that two sibling tasks $p, p' \in Q$ exist so that the allocation in the $(p, p')$-slice mechanism is quasi-bundling or quasi-flipping. Then the mechanism has infinite approximation ratio.
\end{lemma}

\begin{proof} By the characterization result  (Theorem~\ref{theo:addchar} and Observation~\ref{ref:obs-characterization}), the $(p,p')$-slice
  mechanism is a (non task-independent) relaxed affine minimizer.

Consider first the case that the affine minimizer is quasi-flipping (\cite{ChristodoulouKK26} Figure~7 (a)),
with height $\alpha_{p,p'}$ of the flipping boundary (i.e., the height
of the $45^o$ boundary). Lemma~\ref{obs:alphas} implies
$\psi_{p}(s_p)\rightarrow 0$ when $s_p\rightarrow 0,$ and
$\psi_{p'}(s_{p'})\rightarrow 0$ when $s_{p'}\rightarrow 0$, hence for
$(s_p,s_{p'})=(\varepsilon,\varepsilon)$ the allocation is still
quasi-flipping.  In particular, because $0<\psi_p(\varepsilon)\leq
K\varepsilon,$ and $0<\psi_{p'}(\varepsilon)\leq K\varepsilon,$ and the
flipping boundary has height $\alpha_{p,p'},$ exactly one of the tasks
is given to the root for values $(t_p,
t_{p'})=(\alpha_{p,p'}/2,\alpha_{p,p'}/2)$ (\cite{ChristodoulouKK26} Figure~7 (b)). The approximation ratio is
$\mech/\opt>(\alpha_{p,p'}/2)/ 2\varepsilon\rightarrow \infty$ when
$\varepsilon \rightarrow 0.$ Finally, in the quasi-flipping case $\alpha_{pp'}=\infty$ can be excluded by Lemma~\ref{obs:alphas} (i) (to put it simply, both $R_{\{p,p'\}}$ and $R_{\emptyset|\{p,p'\}}$ must be nonempty).

Now consider the case that the relaxed affine minimizer is
quasi-bundling with height of the bundling boundary equal to
$\alpha_{p,p'}$ (\cite{ChristodoulouKK26} Figure~7 (c)). The $2\times 2$ characterization implies that $\alpha_{p,p'}$
is constant (i.e., independent of the $(s,t)$ values when $s$ large enough, see Observation~\ref{ref:obs-characterization}). Set $s_{p}$ so that
$\psi_p(s_p)=\alpha_{p,p'}+\varepsilon/2$ (\cite{ChristodoulouKK26} Figure~7 (d)).  Now for some large $s_{p'}$,
the allocation is almost half-bundling at $p'$, and by the definition
of relaxed affine minimizers, the allocation remains quasi-bundling
when we decrease $s_{p'}.$ If we set $s_{p'}=\varepsilon/(2K)$ and
$t=(\varepsilon, \varepsilon)$, taking into account that
$\psi_{p'}(s_{p'})<K\, s_{p'}=\varepsilon/2$, the existence of the
bundling boundary guarantees that both tasks are given to the leaf (\cite{ChristodoulouKK26} Figure~7 (e)). In
turn this gives unbounded approximation ratio because $\mech\geq
s_p>\psi_p(s_p)/K>\alpha_{p,p'}/K$, while $\opt\leq 2\varepsilon$.

If in the quasi-bundling case $\alpha_{pp'}=\infty$ (e.g., for one-dimensional bundling mechanisms) the argument is analogous. For $s_p=1,\, s_{p'}=\varepsilon/2K$ and  $t=(\varepsilon, \varepsilon),$ both tasks are allocated to the leaves and we obtain unbounded ratio.\end{proof}

\subsection{Base case}
\label{sec:inductionbase}

The base case of the induction proof is the special case of the Box Theorem (Theorem~\ref{thm:box-restated}) for two leaves. In this particular case, multiplicity 2 is sufficient. Theorem~\ref{thm:basecase} establishes that every pair of leaves with two tasks per leaf contains a box of size two. This immediately suggests that there are many boxes of size two. Theorem~\ref{cor:b2} makes it precise.

 \begin{theorem} \label{thm:basecase}(\cite{ChristodoulouKK26} Theorem 41.)
   For a mechanism with bounded approximation ratio, every standard instance with 2 leaves and 2 tasks per leaf contains a box of 2 leaves.   
 \end{theorem}

\begin{theorem}\label{cor:b2}(\cite{ChristodoulouKK26} Theorem 45.) Let $T$ be a standard instance. If the approximation ratio is bounded, a random star of two leaves is not a box with probability at most $2/\sqrt{\ell}$. 
\end{theorem}

\subsection{Induction step}
\label{sec:induction_step}

For the induction step, we consider a star of $k\geq 3$ tasks, which
we call \emph{\slicedoff\ box}, such that all its subsets of $k-1$
tasks are boxes. The precise structure of a \slicedoff\ box is
detailed in the following
definition. %

\begin{definition}[\Slicedoff\ box] \label{def:potentially-good} Fix a
  mechanism of approximation factor less than $K.$ A star $P=\{p_1,\ldots,p_k\}$ from a set of leaves $\cal
  C,$ for $|\mathcal C|=k\geq 3$ is called a \emph{\slicedoff\ box for
    an instance $T$} if $T=(t,\bar s)$ is standard for $P$ and the
  following conditions hold
  \begin{itemize}
  \item for every $i\in [k]$, $P_{-i}=P\setminus \{p_i\}$ is a box for $T$;
  \item for every $q= 1,\ldots,4K/\nu,$ such that $\bar s_{p_k}>q \nu/(4K),$ the instance that results from $T$ when we replace the values of task $p_k$ with $[t_{p_k}=0, \, \bar s_{p_k}-q \nu/(4K)]$, the  $P_{-k}$ is a box. 
  \end{itemize}
  If a \slicedoff\ box is not a box itself, it will be called \emph{\slicedoffx\ box}.
\end{definition}

The definition of a box is about the allocation area $R_P$, in which
all tasks are allocated to the root. Recall that a box is a set of
tasks for which $R_P$ is (almost) an orthotope (a hyperrectangle). On
the other hand, the first condition in the definition of a \slicedoff\
box essentially says that it is a box from which a simplex was cut off
by a diagonal cut (i.e., by a hyperplane of the form $\sum_{i \in [k]}
t_{p_i} = c_{[k]}$). While the definition allows the removed simplex
to be empty (in which case $R_P$ is simply a box), when we want to emphasize
that the simplex is not empty, we call it a \slicedoffx\ box (see
Figure~\ref{fig:box} (b) and (d)).

In the rest of this section we establish that the probability that a
\slicedoff\ box is not a box, is small. We consider a standard
instance $T=(t,s)$ for a set of leaves $\mathcal C=(Q_1,\ldots,Q_k)$, and a star $P=\{p_1,\ldots, p_k\}$ from $\mathcal C$, and we assume that $P$ is a
\slicedoffx\ box. Let $p_k'$ be a random sibling of $p_k$, i.e.,
another task of leaf $Q_k.$ The main result of this section
establishes that {\em almost all other sets $(P_{-k},p_k')$ are
  boxes}.  In particular, we will show via a probabilistic argument
that $(P_{-k},p_k')$ is a box with probability at least
$1-2nK/\ell\xi$ (Lemma~\ref{lemma:many-good-sets}).  In order to show
this result, we consider the $(p_k,p_k')$-slice mechanism which is a
mechanism between two players and two tasks, utilizing the $2\times 2$
characterization (Theorem~\ref{theo:addchar} and
Observation~\ref{obs:possibleFigures}).

We will focus on input points that lie on a specific area around the
allocation region $R_P$, for which we can establish useful properties
when we use the $2\times 2$ characterization. In what follows, two important input
points are $T_\nu(P)$ and $T_{\nu'}(P)$. Recall that
by Definition~\ref{def:nu}, $T_\nu(P)$ is an instance $(t^\nu, s)$
that agrees with $T$ everywhere, except for tasks in $P,$ where
$t^{\nu}_{p_i}=\alpha_{p_i}-4^{k}\nu.$ We define
$T_{\nu'}(P)=(t^{\nu'}, s)$ similarly, but for $\nu'=\nu/4$. Then
$t^{\nu'}_P$ is a point between $t^\nu_P$ and $(\alpha_{p_i})_{p_i\in
  P}.$ Moreover, $T_\nu(P_{-k})$ is the projection in direction of the
$p_k$-axis of $T_{\nu'}(P),$ i.e., all tasks $p_i$ in $P\setminus p_k.$
 have $t$-values equal to $\alpha_{p_i}-4^{k-1}\nu$. We will
use the short notation $\,T_\nu=T_\nu(P)$ and $T_{\nu'}=T_{\nu'}(P).$

We proceed with some useful properties of \slicedoff\ boxes
(Section~\ref{sec:ind-general-observations}), and then we examine the
different possibilities for a $(p_k,p_k')$-slice mechanism
(Section~\ref{sec:ind-p_k-p'k-slice}). Then we conclude with the main
result of this section, that shows that the probability that a \slicedoff\ box is not a
box, is small (Section~\ref{sec:ind-main-lemma}).

\subsubsection{General observations.}
\label{sec:ind-general-observations}
We will need a few simple observations about weakly monotone mechanisms. 
First of all, the following proposition and its corollaries provide some intuition about \slicedoff\ box sets.

\begin{proposition}\label{prop:tech1}(\cite{ChristodoulouKK26} Proposition 47.) Suppose that
  $P=\{p_1,p_2,\ldots,p_k\}$ is a \slicedoffx\ box for $T=(t,\bar s).$
  Consider an arbitrary instance $T'=[(t'_P, t_{[m]\setminus P}) ,\bar
  s]$ such that $t^\nu_{P}< t'_P < t^{\nu'}_P$ coordinate-wise --- and
  so, $T_\nu\leq T'\leq T_{\nu'}.$ Then the point $t'_P$ obeys every
  linear constraint defining the region $R^T_P,$ except for a single
  violated non-redundant constraint $\sum_{i=1}^k t_{p_i}\leq c$ for some
  $c\leq \sum_{i=1}^k \alpha_{p_i}-k\cdot 4^k\cdot\nu,$ which defines
  a bundling boundary facet of $R_P.$
\end{proposition}

\begin{corollary}\label{prop:tech2}(\cite{ChristodoulouKK26} Corollary 48.) Let the conditions be like in Proposition~\ref{prop:tech1}. Then
\begin{itemize}
\item[(i)] the region $R_P$ has a non-redundant bundling facet of equation $\sum_{i=1}^k t_{p_i}=c$ for some $c;$ 

\item[(ii)] for $T'$ all tasks $p_i\in P$ are allocated to the leaves;

\item[(iii)] for each  $j\in [k],$ if in $T'$ we change the $t$-value of only task $p_j$ to $0,$ then all tasks $i\in P$ are allocated to the root.

\item[(iv)] in $T'$ for the task $p_k$ and its  critical value $\psi_{p_k}=\psi_{p_k}(t'_{-p_k},\bar s),$ it holds that the point  $(t'_{P\setminus\{p_k\}}, \psi_{p_k})$ is a point of the bundling boundary facet of $R_P,$ and $\psi_{p_k}>0.$ 

\end{itemize}
\end{corollary}

The next propositions are about points on the bundling boundary facet
of $R_P.$ We consider inputs $T_1=[(t^1_P, t_{[m]\setminus P}) ,\bar
s]$ and $T_2=[(t^2_P, t_{[m]\setminus P}) ,\bar s]$ that differ from
$T_{\nu'}$ and $T_{\nu}$ only in their coordinates $t_{p_i}$ for tasks
$p_i\in P,$ such that $t^\nu_{p_i}\leq t^1_{p_i}< t^2_{p_i}\leq
t^{\nu'}_{p_i}$. All other $t$-coordinates $j\in [m]\setminus P$ stay
the same, i.e.,  $t^\nu_{j}= t^1_{j}= t^2_{j}= t^{\nu'}_{j},$ as well
as the $s$-coordinates.

Due to the relative position of $T_1$ and $T_2,$ and since the
critical value points for task $p_k$ are on a bundling boundary of
$R_P,$ (Corollary~\ref{prop:tech2} (iv)), it is intuitively clear that
the Lipschitz property (Lemma~\ref{prop:lipschitz}) for any two
such points is fulfilled with equality:

\begin{lemma}\label{prop:psiAreFar}(\cite{ChristodoulouKK26} Lemma 49.) Let $T_\nu \leq T_1\leq T_2 \leq T_{\nu'}$  coordinate-wise, and $t^1_{p_i}<t^2_{p_i}$ hold with strict inequality for every $p_i\in P.$ Then  $$\psi_{p_k}(t_{-p_k}^1,\bar s)-\psi_{p_k}(t_{-p_k}^2,\bar s)=|t_{-p_k}^1-t_{-p_k}^2|_1=\sum_{i\in [k-1]} (t_{p_i}^2-t_{p_i}^1).$$
  \end{lemma}

The proof of the previous proposition shows 
$\psi_{p_k}(t^1)-\psi_{p_k}(t^2)=|t^1-t^2|_1$ essentially using the
fact that for $i=1,2$ the points $(t^i_{-p_k},\psi_{p_k}(t^i))\in
\mathbb R^k$ are on the bundling boundary of $R_P$ and
$R_{\emptyset|P}.$ The next proposition and its corollary show that
the reverse also holds: If
$\psi_{p_k}(t_{-p_k}^1)-\psi_{p_k}(t_{-p_k}^2)=|t_{-p_k}^2-t_{-p_k}^1|_1,$
then for all points $t'$ with $t^1\leq t'\leq t^2$ their critical
inputs for task $p_k$ are on a bundling boundary of $R_P$.

\begin{proposition}\label{prop:FarPsiAreBundling}(\cite{ChristodoulouKK26} Proposition 50.) Let $s$ be fixed, let
  $t^1_{p_i}< t^2_{p_i}$ for every $p_i\in P,$ and $t^1_{j}=t^2_j$ for
  $j\in [m]\setminus P.$ If
  $\psi_{p_k}(t_{-p_k}^1)-\psi_{p_k}(t_{-p_k}^2)=\sum_{i\in [k-1]}
  (t_{p_i}^2-t_{p_i}^1),$ then
\begin{enumerate}
\item[(i)]  for  $(t^1_{-p_k},\psi_{p_k}(t^1)-\varepsilon)$ all tasks $p_i\in P$ are given to the root for $0<\varepsilon< \psi_{p_k}(t^1),$ and 
\item[(ii)] for  $(t^2_{-p_k},\psi_{p_k}(t^2)+\varepsilon)$ all tasks $p_i\in P$ are given to the leaves  for $\varepsilon>0.$
\end{enumerate}
    
  \end{proposition}

\begin{corollary}\label{cor:FarPsiAreBundling}(\cite{ChristodoulouKK26} Corollary 51.) Assume that $t^1_{p_i}< t^2_{p_i}$ for every $p_i\in P,$ and $t^1_{j}=t^2_j$ for $j\in [m]\setminus P,$ and $\psi_{p_k}(t_{-p_k}^1)-\psi_{p_k}(t_{-p_k}^2)=\sum_{i\in [k-1]} (t_{p_i}^2-t_{p_i}^1).$ 
Let $t^1\leq t'\leq t^2$  be so that $t^1_{p_i}< t'_{p_i}< t^2_{p_i}\, (i\in [k]).$ Then
\begin{enumerate}
\item[(i)] for  $(t'_{-p_k},\psi_{p_k}(t')-\varepsilon)$ all tasks $p_i\in P$ are given to the root, and 
\item[(ii)] for  $(t'_{-p_k},\psi_{p_k}(t')+\varepsilon)$ all tasks $p_i\in P$ are given to the leaves. 
\end{enumerate}

\end{corollary}

\subsubsection{The $p_k$-$p_k'$ slice}
\label{sec:ind-p_k-p'k-slice}

We now consider a \slicedoffx\ box $P=\{p_1,\ldots, p_k\}$ and choose
a random sibling $p_k'$ of $p_k$.  We will consider all different
possibilities for the ($p_k, p_k'$)-slice mechanism based on the
$2\times 2$ characterization. First, we consider the case that the
mechanism is task-independent, establishing that then $(P_{-k}, p'_k)$
must be a box (Lemma~\ref{lem:crossing}). Then we treat the non
task-independent case, excluding almost surely affine minimizers and constant mechanisms
(Lemma~\ref{lemma:linearBoundary}) and bounding from above the number of
$p_k'$ siblings for which the mechanism can be relaxed-affine
minimizer (Lemma~\ref{lem:noHalfBundling}).%

\paragraph{The case of task-independent (crossing) allocations.}

In this paragraph we treat the case when in at least two points $T_1<T_2$ between $T_\nu$ and $T_{\nu'}$ the allocation of  the $(p_k, p_k')$-slice is crossing. We prove that in this case the set $(P_{-k}, p_k')$ is a box.

\begin{lemma}\label{lem:crossing}(\cite{ChristodoulouKK26} Lemma 52.) Suppose that $P=\{p_1,p_2,\ldots,p_k\}$ is a \slicedoffx\ box for $T=(t,\bar s)$. Fix a sibling $p'_k$ of $p_k.$ Assume that there exist two distinct instances $T_1,\,T_2,$ such that 
\begin{itemize}

\item $T_\nu \leq T_1\leq T_2 \leq T_{\nu'}$  coordinate-wise, and $t^\nu_{p_i}<t^1_{p_i}<t^2_{p_i}< t^{\nu'}_{p_i}$ holds with strict inequality for every $p_i\in P;$
 \item for both $T_1$ and $T_2$, in the $(p_k, p_k')$ slice mechanism  the allocation of the root is crossing.

\end{itemize}
Then $(P_{-k}, p'_k)$ is a box for $T.$

\end{lemma}

\paragraph{The case of non-crossing allocations.}

In the remaining part, we need to exclude the case that between
$T_\nu$ and $T_{\nu'}$ there are ``many'' points where the allocation of
the slice mechanism $(p_k, p_k')$ is non-crossing.  Allocations of
$2\times 2$-mechanisms that are non-crossing for the root, occur (1)
in the linear part (linear as function of $s_{p_k}$) of relaxed affine
minimizers (including constant mechanisms), or (2) they must be half-bundling at
$p_k$ (including completely bundling). %

The following paragraphs deal with both of these cases; for any given
point $T,$ case (1) --- linear boundary functions --- will be excluded
to occur in $k$ different points, positioned appropriately, hence
linear functions are excluded almost surely; case (2) is excluded with
high probability over the choice of $p_k'$.

\paragraph{Case 1: Linear boundary functions.}

Here we exclude the case that in $k$ carefully selected points $T_0,
\, T_1, \ldots , T_{k-1}$ the $(p_k, p_k')$-slice mechanism is a
(relaxed) affine minimizer, and so that in each of these points the
critical value function $\psi_k(s_{p_k})$ is truncated linear. This
implies that linear boundary functions are excluded almost surely. The next lemma is an adaptation of (\cite{ChristodoulouKK26} Lemma 53.). We omit the proof of those claims where no change was needed.

\begin{lemma} \label{lemma:linearBoundary} Assume that the
  approximation ratio is less than $K.$ Suppose that
  $P=(p_1,\ldots,p_k)$ is a \slicedoffx\ box for standard instance
  $T=(t,\bar s).$ Fix a sibling $p_k'$ of $p_k.$ Finally let
  $h=(5/8)\cdot 4^{k}\nu,$ and let $0<\eta< \nu/4n$ be some small
  number. There exist no distinct instances $T_0=(t^0,\bar s)$,
  $T_1=(t^1,\bar s),\ldots, T_{k-1}=(t^{k-1},\bar s)$ such that
  \begin{itemize}
  \item $T_\nu\leq T_j\leq T_{\nu'},$ and $\,t^\nu_P<t^j_P<t^{\nu'}_P\,$ for every $j = 0, \ldots, k-1;$ 
  \item $t^0_{p_j}+h<t^j_{p_j}<t^0_{p_j}+h+\eta$ for all $j\in [k-1];$
  \item $t^0_{p_i}<t^j_{p_i}<t^0_{p_i}+\eta$ for all $j\in [k-1],$ and all $i\neq j;$
  \item the boundary function $\psi_{p_k}$ is truncated linear in $s_{p_k},$ that is,
    \begin{align*}
      \psi_{p_k}(t^j_{-p_k},\bar s)=\max(0\,,\, \lambda(t^j_{-p_k},\bar s_{-p_k}) s_{p_k} - \gamma(t^j_{-p_k}, \bar s_{-p_k})),
    \end{align*}
    for  each of the instances $T_j,\quad$ $j = 0, 1,\ldots, k-1.$ 
  \end{itemize}
\end{lemma}

\begin{proof}
Towards a contradiction, suppose that such instances $T_0, T_1,\ldots , T_{k-1}$ exist. For simplicity of notation, throughout the proof we write $\psi_k^j (s_{p_k})$ for the critical value function $\psi[T_j]_{p_k}(s_{p_k}),$ defined for the fixed values $t_{-p_k}^j$ and  $\bar s_{-p_k}.$ 

The high-level idea is the following (see Figure~5 in \cite{ChristodoulouKK26}). Let us ignore all other
dimensions, except for those $k$ coordinates of tasks $P.$ Due to
Corollary~\ref{prop:tech2} ((ii), (iii) and (iv) respectively), we know
that in the allocation of $T_j,$ $j=0,\,1,\ldots, k-1$, all tasks
are given to the leaf, when we change the $t_{p_k}$-value to $0,$ all
tasks $p_i\in P$ are given to the root, and the respective boundary
points $(t^j_{-p_k},\psi_k^j(\bar s_{p_k}))\in \mathbb R^k$ are on a
bundling boundary facet of $R_P.$ Now, as we reduce $s_{p_k},$ due to
the linearity of each boundary function $\psi^j_k(s_{p_k}),$ these $k$
boundary points move by the same speed $\lambda$ towards the
coordinate-plane $t_{p_k}=0,$ so that all the time they remain on a
bundling boundary.  For some positive $s^*_{p_k}$ this bundling
boundary facet will reach the plane $t_{p_k}=0,$  %
and there it will `form' a bundling boundary of dimension $k-1$ for the task-set $P_{-k}.$ But
this will prove that for this $s^*_{p_k}$ the set $P_{-k}$ cannot be
a box, contradicting the fact that $P$ is \slicedoff\ box.

We start with some observations about the critical values $\psi_k^j
(\bar s_{p_k}),$ and the respective boundary points of $R_P$ when
$s_{p_k}=\bar s_{p_k}$.  Note that for every $j\in [k-1]$ by
definition $ t^0_{p_i}< t^j_{p_i}$ holds for all $p_i\in P.$
Therefore, by Lemma~\ref{prop:psiAreFar} $$\psi^0_k (\bar
s_{p_k})-\psi_k^j (\bar s_{p_k})=|t_{-p_k}^0-t_{-p_k}^j|_1=\sum_{i\in
  [k-1]}(t^j_{p_i}-t^0_{p_i}).$$ Moreover, for every $j\in [k-1]$ this
difference roughly equals $h,$ in particular
$$\, h<\sum_{i\in [k-1]}(t^j_{p_i}-t^0_{p_i})<h+n\cdot \eta.$$ 

Therefore, all these critical values are almost the same, that is, for
any pair of indices $j,j'\in[k-1]$
it holds $$|\psi_k^j (\bar s_{p_k})- \psi_k^{j'} (\bar s_{p_k})|<
n\cdot\eta.$$ Next we show that these relative positions remain the
same {\em for every $s_{p_k}.$} %

\begin{claim*}[i]

The boundary functions for $t^j$, $j=0,\,1,\ldots , k-1$,
  \begin{align*}
      \psi_k^j (s_{p_k}) & =\lambda(t_{-p_k}^j) s_{p_k} - \gamma(t_{-p_k}^j)
  \end{align*}
  have the same linear coefficient, i.e., $\lambda(t_{-p_k}^0)=\lambda(t_{-p_k}^1)= \ldots =\lambda(t_{-p_k}^{k-1})=\lambda$. 
\end{claim*}

As a corollary we obtain that the difference between (positive) critical values $\psi_k^0$ and $\psi_k^j$ remains the same \emph{for every} $s_{p_k}$.

  \begin{claim*}[ii] Let $j\in [k-1].$ \emph{For every} $s_{p_k}$, for which $\psi_k^j ( s_{p_k})>0,$ the following hold:  
\begin{enumerate}
\item[(a)]  $ \psi_k^0 (s_{p_k})-\psi_k^j (s_{p_k})=|t_{-p_k}^0-t_{-p_k}^j|_1;$
  
\item[(b)]  $|\psi_k^j (s_{p_k})- \psi_k^{j'} ( s_{p_k})|\leq n\cdot \eta$ for every $j'\in [k-1];$ 
\item[(c)] the point $(t^j_{-p_k}, \psi_k^j(s_{p_k}))\in \mathbb R^k$ is on a bundling boundary, in particular on a boundary of $R^{T'}_P,$ where $T'$ is obtained from $T$ by replacing $\bar s_{p_k}$ by $s_{p_k},$ for $j=0, 1, \ldots, k-1;$
\end{enumerate}
\end{claim*}

Let $\psi^1_k(s_{p_k}) = \max(0, \lambda\, s_{p_k} - \gamma).$ Next we prove lower and upper bounds for $\gamma$ and  $\lambda.$

\begin{claim*}[iii] Let $\psi^1_k(s_{p_k})=\max(0, \lambda\, s_{p_k} -
  \gamma)$ for some $\lambda$ and $\gamma$ that do not depend on
  $t_{p_k}$
  and $s_{p_k}$. Suppose further that the approximation
  ratio is less than $K,$ then 
  \begin{enumerate}
\item[(a)] $0\leq  \gamma < \lambda; $
\item[(b)] $\lambda<2K.$ 
\end{enumerate}
\end{claim*}

\begin{proof} (a) By Corollary~\ref{prop:tech2} (iv) holds that $\psi^1_k (\bar s_{p_k})>0$. Now notice that $\lambda\cdot 1-\gamma=\psi^1_k(1)\geq \psi^1_k(\bar s_{p_k}) > 0.$ This implies $\gamma< \lambda$. 

  Next we show $\gamma\geq 0.$ We claim that $\psi_k^1(s_{p_k}=0)=0.$
  Assume the contrary, that $\psi_k^1(s_{p_k}=0)>0.$ Then, by Claim
  (ii) (c), the $(t_{-p_k}^1\,, \, \psi_k^1(0))$ is a point with
  positive coordinates on a bundling boundary of $R_P^{T'}$ where $T'$
  is obtained from $T$ by setting $s_{p_k}=0.$ On the other hand, then
  by Lemma~\ref{obs:Rnonempty} $s_{p_k}>0$ must hold,
  contradiction. So we have $\psi_k^1(s_{p_k}=0)=0.$ Consequently,
  $\lambda\cdot 0 - \gamma\leq 0,$ so $0\leq \gamma.$ Notice that part (a) of this claim excludes constant mechanisms with $\lambda=0.$

(b) Assume that $\lambda\geq 2K.$ We set $s_{p_k}=2,\, t_{-p_k}=t^1_{-p_k}$ and $t_{p_k}=\psi_k^1(s_{p_k})-\varepsilon=\lambda s_{p_k}-\gamma-\varepsilon.$ Then task $p_k$ is allocated to the root, and the makespan is $\mech\geq \lambda s_{p_k}-\gamma-\varepsilon \geq \lambda \cdot 2-\lambda-\varepsilon= \lambda -\varepsilon\geq 2K-\varepsilon.$

The contribution of the tasks $[m]\setminus P$ to the optimal makespan is $0,$ since the other leaves are trivial. We have $\opt\leq \max_i s_{p_i}=s_{p_k}=2.$ We obtain $\mech/\opt\rightarrow K$ when $\varepsilon \rightarrow 0,$ which contradicts our assumption about the approximation factor.
\end{proof}

\begin{claim*}[iv]
  There exists $q^*\in \{0,1,\ldots,4K/\nu\}$ such that $\bar s_{p_k}\geq q^*  \nu/(4K),$ and $\psi^j_k(\bar s_{p_k}-q^*  \nu/(4K))$ is in the interval $(0, \nu)$, for every $j\in [k-1].$ 
\end{claim*}
\begin{proof} By Corollary~\ref{prop:tech2} (iv) it holds for every
  $j\in[k]$ that $\psi^j_k(\bar s_{p_k})>0.$ If also $\psi^j_k(\bar
  s_{p_k})<\nu$ holds for every $j\in[k-1],$ then the claim holds
  trivially with $q^*=0.$
  
  Otherwise let $\bar q\leq 4K/\nu$ be the largest integer such that
  $\bar s_{p_k}\geq \bar q \nu/(4K)$ (recall that $\bar s_{p_k}\leq
  1$). Consider the sequence of values $\psi^1_k(\bar s_{p_k}-q
  \frac{\nu}{4K}),$ for $q=0,1,\ldots,\bar q.$ By Claim (iii),
  $\lambda< 2K$, so successive values in this sequence differ by at
  most $\lambda (\nu/4K)\leq 2K\nu/4K=\nu/2.$ Moreover, because
  $\gamma\geq 0,$ we have $\psi^1_k(0)=0.$ Therefore for some $0\leq
  q^*\leq \bar q,$ the $\psi^1_k(\bar s_{p_k}-q^*\cdot\frac{\nu}{4
  K})$
  falls between $\nu/4$ and $3\nu/4.$ Finally, by Claim (ii) (b)
  it follows that $\psi^j_k(\bar s_{p_k}-q^*\cdot\frac{\nu}{4K})\in (0,
  \nu)$ for \emph{every} $j\in[k-1],$ given that $n\eta< \nu/4.$
\end{proof}

Let $s^*_{p_k}:=\bar s_{p_k}-q^* \nu/(4K)$ be the value from the
previous claim for which $\psi^j_k(s^*_{p_k}) \in (0, \nu)$, for every
$j\in[k-1].$ For the rest of the proof we fix $(s^*_{p_k}, \bar
s_{-p_k}).$ We know from Claim (ii) (c) that for every $j\in [k-1],$
in an arbitrarily small neighborhood of the point $(t^j_{-p_k},
\psi^j_k(s^*_{p_k}))\in \mathbb R^k$ there are points $t_P$ so that
all tasks $p_i\in P$ are allocated to the root, and also there are
points so that none of the tasks $p_i\in P$ is allocated to the root.

Recall that  $t_{p_i}=0$ for every $p_i\in P$ in the instance $T=(t,\bar s).$ Let $T^*=(t, (s^*_{p_k}, \bar s_{-p_k}))$ denote the instance that we obtain from $T$ by replacing $\bar s_{p_k}$ by $s^*_{p_k}.$  Notice that $T^*$ is standard instance for the star $P_{-k}=P\setminus {p_k}.$ We need to be careful, because all we know about the allocation of instance $T^*$ are the $k$ shifted points $(t^j_{-p_k},
\psi^j_k(s^*_{p_k}))$ of a bundling boundary (in particular, the $\alpha_{p_j}$ values may have changed).

The rest of the proof is that same as in (\cite{ChristodoulouKK26} Lemma 53.), and we omit the details: For every $j\in [k-1]$ we  define
two nearby instances to $(t^j_{-p_k}, \psi^j_k(s^*_{p_k})),$ called
$\underline T_j$ and $\overline T_j$ so that for each of them the
$t_{p_k}$ entry is zero, the $s_{p_k}$ entry equals $s^*_{p_k},$ and
for $\underline T_j$ all tasks $p_i\in P_{-k}$ are given to the root,
and for $\overline T_j$ all tasks $p_i\in P_{-k}$ are given to the
leaves. %
This contradicts the assumption that for $s^*_{p_k}$ the set $P_{-k}$ is a box, and thus also that $P$ is a \slicedoff\ box (the second
bullet of Definition~\ref{def:potentially-good}), which concludes the proof.
 
 \end{proof}

\paragraph{Case 2: Half-bundling boundary functions.}

We consider the case that for some $T_\nu(P)\leq \hat T \leq
T_{\nu'}(P),$ the slice mechanism $(p_k,p_k')$ is non-crossing, and so
that it has a boundary half-bundling at $p_k$ (or even fully bundling
at both $p_k$ and $p_k'$). Recall that this includes one-dimensional
bundling mechanisms as well as relaxed affine minimizers that might
become nonlinear for small $s_{p_k}.$
For arbitrary \emph{fixed} $T_\nu(P)\leq \hat T \leq T_{\nu'}(P),$ we
exclude that the slice mechanism $(p_k,p_k')$ is half-bundling at
$p_k,$ \emph{for many } different siblings $p_k'\in Q_k.$ 

\begin{lemma}\label{lem:noHalfBundling} (cf. \cite{ChristodoulouKK26} Lemma 54.) Suppose that $P=\{p_1,p_2,\ldots ,p_k\}$ is a \slicedoffx\ box for standard instance $T=(t, \bar s),$ and corresponding leaf set $\mathcal C=\{Q_1, Q_2, \ldots Q_k\},$ and the approximation ratio is less than $K.$ Let 
  $T_\nu(P)\leq \hat T \leq T_{\nu'}(P),$ be an arbitrary input. Then
  fewer than $K /\xi$ siblings $p_k'\in Q_k$ of $p_k$ have a $2\times
  2$ slice mechanism $(p_k, p_k')$ in $\hat T$ that is half-bundling
  at $p_k.$
\end{lemma}

\begin{proof} Let $\hat T=((\hat t_P, t_{-P}), \bar s),$ (since by
  definition $t_{-P}=\hat t_{-P}$). We fix and omit $\bar s$ from the
  notation in the rest of the proof.  Let $\psi_{p_k}=\psi_{p_k}[\hat
  T]$ be the critical value of $t_{p_k}$ in $\hat T.$

  Let $p_k'$ be an arbitrary sibling of $p_k,$ such that the slice
  $(p_k,p_k')$ in $\hat T$ is half-bundling at $p_k$ (see
  Figure~14 in \cite{ChristodoulouKK26}). Then, for the input point
  slightly above the boundary $t^+=(\hat t_{P_{-k}},
  t_{p_k}=\psi_{p_k}+\varepsilon, t_{-P})$ neither of the tasks $p_k$ or
  $p_k'$ is allocated to the root. This follows from the definition of
  $\psi_{p_k},$ the fact that $t_{p_k'}=0$ (because $T$ is standard
  instance for $P$), and that the boundary for $p_k$ in the slice
  $(p_k,p_k')$ is locally bundling at $t_{p_k'}=0$ by the definition of
  half-bundling allocations (Definition~\ref{def:half-bundling}).

  Assuming by contradiction that $p_k$ has at least $K /\xi$
  different siblings $p_k'$ with which it is half-bundling, for input
  $t^+$ we obtain that all these $p'_k$ tasks are allocated to the
  same leaf player, incurring a cost of at least $\xi$ for each of
  them (because $Q_k$ is standard leaf, hence $\bar
  s_{p_k'}>\xi$). Therefore the makespan achieved by the mechanism is
  at least $(K /\xi)\cdot \xi = K. $  On the other hand, since
  tasks in $[m]\setminus P$ are trivial, we have that the makespan of the
  optimal allocation is at most $\max_{p_i\in P}\bar s_{p_i}\leq 1.$
  This would imply approximation ratio at least $K,$
  contradiction. \end{proof}

\subsubsection{The main inductive lemma.}
\label{sec:ind-main-lemma}
The results of Section~\ref{sec:ind-p_k-p'k-slice} culminate in the following lemma.

\begin{lemma} \label{lemma:many-good-sets} Suppose that
  $P=(p_1,\ldots,p_k)$ is a \slicedoffx\ box for $T.$ Then %
  the number of different siblings $p_k'$ of $p_k$ for which 
$(P_{-k},p_k')$ is a box for $T$,  is at least $\ell-2nK/\xi.$
\end{lemma}

\begin{proof}
  Suppose that $P$ is not a box, and take a random sibling $p_k'$ of $p_k.$
  For $j=0, 1, \ldots k-1,$ we fix  instances $T_j, T_j^* $  such that 
  \begin{itemize}
  
  \item[(i)] $T_\nu \leq T_j\leq T_j^*\leq T_{\nu'}$ and all $t$-coordinates of $T_j, \,T_j^*$ are rational;

  \item[(ii)] $t^{j*}_{p_i}=t^{j}_{p_i}+\varepsilon$ for some small $\varepsilon$ for every $i\in[k];$
  \item[(iii)] the relative position of $T_0, T_1, \ldots T_{k-1}$ is as defined in the statement of Lemma~\ref{lemma:linearBoundary}, moreover this property holds even if we replace any subset of the $T_j$ by the respective $T_j^*$ instances.
  \end{itemize}
  Observe first that such instances exist:  in  Lemma~\ref{lemma:linearBoundary} $h=(5/8)\cdot 4^{k}\nu,$ and $\eta$ is arbitrarily small; on the other hand, $t^{\nu'}_{p_i}-t^{\nu}_{p_i}=(4^k-4^{k-1})\nu=(3/4)4^k\nu>(5/8)4^k\nu.$ So there is enough space for instances in the prescribed relative position.

  By Lemma~\ref{lem:noHalfBundling} and using the union bound, the
  number of siblings $p_k'\in Q_k\setminus \{p_k\}$ for which the
  $(p_k,p_k')$ slice mechanism is half-bundling in at least one of
  these $2k$ instances, is less than $2kK /\xi\leq 2nK/\xi-1.$
  Therefore, most of the siblings $p'_k$ (at least $\ell-1-(2nK/\xi-1)=\ell-2nK/\xi$ of them)
  do not have a half-bundling boundary with $p_k$ in any of the
  $2k$ fixed instances.  We focus now on such a sibling $p_k'.$ By
  Lemma~\ref{lemma:linearBoundary}, there must exist at least one
  $j\in\{0,1,2,\ldots, k-1\}$ such that in both $T_j$ and $T_j^*$ the
  $\psi[T_j]_{p_k}(s_{p_k})$ and $\psi[T_j^*]_{p_k}(s_{p_k})$
  functions are not truncated linear. Altogether, this excludes
  (non task-independent) relaxed affine minimizer mechanisms,
  constant, and one-dimensional bundling mechanisms as well as
  discontinuities in these two points. The remaining possibility is
  that the allocation of the slice $(p_k,p_k')$ is crossing for both
  $T_j$ and $T_j^*.$ Then, by Lemma~\ref{lem:crossing} the star
  $(P_{-k}, p_k')$ is a box for $T.$\end{proof}

\subsection{Many boxes of size $n-1.$}
\label{sub:existence}
We are now ready to show the main result of this section
(Corollary~\ref{cor:bound-on-b}) that proves %
that a random star of size $n-1$ is a box with arbitrarily high probability.
The proof is by induction on the size of the box, and the
next definition is handy to define `bad' events on $k$ leaves.

\begin{definition}[Probability $b_k$] Fix a mechanism, and let $1\leq k\leq n-1.$ Suppose that $T$ is a standard instance for a set $\mathcal C$ of $k$ leaves. We denote by $b(T,\mathcal C)$ the probability that a random star $P$ of $k$ tasks from $\mathcal C$ is not a box. Let $b_k$ denote the supremum of all possible $b(T,\mathcal C)$ for all choices of $T$ and $\mathcal C$ with $|\mathcal C|=k$ for the given mechanism.
\end{definition}
In the next lemma we upper bound the probability that a random star $P$ of $k$ tasks is not a \slicedoff\ box, in terms of  $b_{k-1}.$
The lemma afterwards finally combines the obtained probabilities.

\begin{lemma} \label{lemma:potentially-good} Fix any instance $T$
  which is standard for a set $\cal C$ of $k\geq 3$ leaves, and let
  $P$ be a random star of $k$ tasks from $\cal C$. The
  probability that $P$ is not a \slicedoff\ box for $T$ is at most $(5K/\nu-1)b_{k-1}.$
 
\end{lemma}

\begin{proof} Take a random star $P=\{p_1, p_2,\ldots , p_k\}$ of $k$ tasks from $\mathcal C.$ For each $i\in [k]$ the probability that $P_{-i}$ is not a box for $T$ is at most $b_{k-1}.$ Also the probability that $P_{-k}$ is not a box when we set task $p_k$ to $(t_{p_k}=0, s_{p_k}=\bar s_{p_k}- q\nu/(4K)),$ is at most $b_{k-1}$ for (at most) each of  $q= 1, \ldots , 4K/\nu. $ By the union bound, the probability that some of these bad events happens is at most $(k+4K/\nu)b_{k-1}< (n+4K/\nu)b_{k-1}<(5K/\nu-1)b_{k-1}$ (given that $n<K-1$ w.l.o.g., and $\nu$ is much smaller than 1).
\end{proof}

\begin{lemma} \label{lemma:recurrence}
  Fix a mechanism with approximation ratio less than $K,$ then for $3\leq k\leq n-1$ $$b_k\leq \left (\frac{5K}{\nu}-1 \right )b_{k-1}+ \frac{2nK}{\xi \sqrt{\ell} }.$$ 
\end{lemma}

\begin{proof} Let $T$ be a standard instance for some set of leaves $\mathcal C=\{Q_1,\ldots , Q_k\},$ and let $\mathcal C_{-k}=\{Q_1, \ldots Q_{k-1}\}.$ Consider a star $P^*=\{p_1,\ldots, p_{k-1}\}$ from $\mathcal C_{-k}.$ %
  Let's call the sets $P^*\cup\{p_k\}$, $p_k\in Q_k$, extensions of $P^*$.
  
  Let $0\leq y\leq \ell$ be the number of the extensions of $P^*$ that are
  \slicedoff\ boxes. How many of these extensions are boxes? Either all $y$
  extensions are boxes, in which case the number of box extensions is at least
  $y$ or some extension is not a box and $\ell-2nK/\xi$ other extensions are boxes
  (Lemma~\ref{lemma:many-good-sets}). Therefore the number of extensions that
  are boxes is at least  $\min\{y, \ell-2nK/\xi\}\geq y (\ell-2nK/\xi)/\ell= y (1-2nK/\ell\xi).$ Since this holds for every star $P^*$, the probability that a random star of size $k$ is a box is at
  least $1-2nK/\xi\ell$ times the probability that a random star of size
  $k$ is a \slicedoff\ box.
  
  For a rigorous proof, let ${\mathcal R}$ be the set of different stars of $k-1$ tasks from $\mathcal C_{-k}.$ 
We select a random star $P$ of $k$ tasks from $\mathcal C,$ and define the events:
$\quad Y:\,$ "$P$ is a \slicedoff\ box"; $\quad X:\,$  "$P$ is a box"; $\quad E_{P^*}:\,$  "$P_{-k}=P^*\,$".  Note that  in general $X\not\subset Y.$
By the above argument $\mathbb P(X|E_{P^*})\geq (1-2nK/\xi\ell)\cdot \mathbb P(Y|E_{P^*}).$ Therefore we obtain 
\begin{align*}\mathbb P(X)=\sum_{P^*\in \mathcal R} \mathbb P(X| E_{P^*})\mathbb P(E_{P^*})\\
\geq (1-2nK/\xi\ell) \sum_{P^*\in \mathcal R} \mathbb P(Y|E_{P^*}) \mathbb P(E_{P^*})\\
=(1-2nK/\xi\ell)\cdot \mathbb P(Y).\end{align*}
By Lemma~\ref{lemma:potentially-good}, the
  probability that a random star of size $k$ is a \slicedoff\ box is at
  least $1-(5K/\nu - 1) b_{k-1}; $ so we get $$\mathbb P(X)\geq \left (1-\frac{2nK}{\xi\ell}\right )\left [1-\left (\frac{5K}{\nu} - 1\right ) b_{k-1}\right ]> 1-\frac{2nK}{\xi\ell}-\left (\frac{5K}{\nu} - 1\right ) b_{k-1}.$$  
This further yields $$b_k=1-\mathbb P(X)< \left (\frac{5K}{\nu} - 1\right ) b_{k-1}+ \frac{2nK}{\xi\ell} < \left (\frac{5K}{\nu} - 1\right ) b_{k-1}+ \frac{2nK}{\xi\sqrt{\ell}}.$$
\end{proof}

\begin{corollary} \label{cor:bound-on-b}
  Fix a mechanism with approximation ratio less than $n,$ then for $2\leq k\leq n-1$ $$b_k\leq \left (\frac{5K}{\nu}\right )^{k-2}\cdot \frac{2nK}{\xi \sqrt{\ell} }.$$ 
\end{corollary}
\begin{proof} For $k=2,$ using Theorem~\ref{cor:b2},  $b_2\leq 2/\sqrt{\ell}< 2nK/(\xi\sqrt{\ell}),$ so the base case holds. Now for $k\geq 3$ by induction we obtain
\begin{align*}
b_k & \leq \left (\frac{5K}{\nu}-1 \right )b_{k-1}\,+ \,\frac{2nK}{\xi \sqrt{\ell} }\\
& \leq 
\left (\frac{5K}{\nu}\,-\, 1 \right )\left (\frac{5K}{\nu}\right )^{k-3}\cdot \frac{2nK}{\xi \sqrt{\ell} }\,+ \,\frac{2nK}{\xi \sqrt{\ell} }\\
& = 
\left (\left (\frac{5K}{\nu}\right )^{k-2}\,-\,\left (\frac{5K}{\nu}\right )^{k-3}\,+\,1 \right )\cdot \frac{2nK}{\xi \sqrt{\ell} }\\
& \leq  
\left (\frac{5K}{\nu}\right )^{k-2}\cdot \frac{2nK}{\xi \sqrt{\ell} }.%
\end{align*}
\end{proof}

Setting $k=n-1,$ and using $\nu<\xi,$ and $n+1<K,$ we obtain the following rough bound.%

\begin{corollary} 

$b_{n-1} \leq  
\left (\frac{5K}{\nu}\right )^{n-3}\cdot \frac{2nK}{\xi \sqrt{\ell} }< \left (\frac{5K}{\nu}\right )^{n-1}\cdot \frac{1}{ \sqrt{\ell} },$ so 
for every instance $T=(t,\bar s)$ standard for $\{C_i\}_{i\in [n-1]},$ a random star of $n-1$ tasks is a box with arbitrarily high probability, if the edge-multiplicity $\ell$ is high enough. 
\end{corollary}

	\section{Omitted Details from~\Cref{sec:upperbound}}
\label{sec:appendix-upper-bound}

\subsection{Payment Scheme}
\label{sec:payment-scheme-upperbound}

Given any instance $t$ for \textsc{Exp-Bounded-Square}, we define payments $p_{ij}$ for each $i \in N$, $j \in M$.
The overall payment that machine $i$ receives is then given by $\mathcal P_i(t) := \sum_{j \in M} p_{ij}$, which defines a possible payment scheme such that the mechanism is universally truthful.

Consider any fixed task $j \in M$, and suppose $j$ is allocated to machine $i^*$.
Set $p_{ij} = 0$ for all other machines $i \neq i^*$, fix their reported processing times, and define
\[
	t_{-i^*}^{\min} := \min \{\; t_{ij} \;\vert\; i \in N,\;i \neq i^* \} \enspace.
\]
If $t_{-i^*}^{\min} = 0$, we set $p_{i^*,j} = 0$.
Otherwise, we set
\[
	p_{i^*,j} = \sup \left\{z \ge 0 \;\colon\; E_{i^*j} \cdot z^2 \le E_{ij} \cdot (t_{ij})^2 \text{ holds for all } i \neq i^* \text{ with } t_{ij} \le 2 \min \{z, t_{-i^*}^{\min} \} \right\} \enspace.
\]
Thus, $p_{i^*,j}$ is the maximum processing time that $i^*$ could report for task $j$ such that it is still allocated to $i^*$.
Since the allocation algorithm is weakly monotone (see~\Cref{sec:universal-truthfulness}), this payment scheme is universally truthful.

\subsection{Bounding the Expected Load of Each Machine}
\label{sec:expected-load-lu-yu-upperbound}

\begin{proof}[Proof of~\Cref{lem:lu_yu_max_expected_load}]
	Consider a fixed optimal allocation for $t$, and let $J_i$ denote the set of tasks that is allocated to machine $i \in N$ in that optimal allocation.
	For every task $j\in M$, let $t^{\min}_j=\min_{i\in N} t_{ij}$, and let $N_j$ denote the set of eligible machines for task $j$.
	We assume w.l.o.g. that $t^{\min}_j> 0$ for every task $j$.
	For the optimum makespan it holds 
	\[
		\Opt{t} = \max_{i\in N}\sum_{j\in J_i} t_{ij}.
	\]
	Due to \Cref{lem:probabilities-exp-bounded-square}, the expected running time of an arbitrary machine $i$ in an allocation computed by \textsc{Exp-Bounded-Square} for instance $t$ can be expressed by
	\[
		\mathbb E[\theta_i(t)] 
		= \sum_{s \in N} \sum_{j\in J_s} x_{ij}t_{ij}
		= \sum_{j\in J_i} x_{ij}t_{ij} + \sum_{\substack{s\in N,\\ s\neq i}}\, \sum_{j\in J_s} x_{ij}t_{ij},
	\]
	where 
	\[
		x_{ij} = 
			\begin{cases}
				 \frac{1/(t_{ij})^2}{\sum_{s\in N_j}1/(t_{sj})^2} &,~\text{if}~i \in N_j,\\
				 0 &,~\text{otherwise}.
			\end{cases}
	\]
	For the first sum, it holds $\sum_{j\in J_i} x_{ij}t_{ij} \leq \sum_{j\in J_i} t_{ij} \leq \Opt{t}$.
	Next, we show that for every $s\neq i$ it holds that $\sum_{j\in J_s} x_{ij}t_{ij} \leq \Opt{t}/2,$ which will yield 
	\[
		\mathbb E[\theta_i(t)] \leq \left(\frac{n+1}{2}\right)\cdot \Opt{t}
	\]
	as required.
	
		\begin{claim*}
			For every $s\neq i$ and every $j\in J_s$ holds that $x_{ij}t_{ij}\leq \frac{t_{sj}}{2}.$
		\end{claim*}
		\begin{proof} 
			If $i \notin N_j$ then $x_{ij}=0,$ and the claim holds. 
			Otherwise $i\in N_j,$ and so $t_{ij}\leq 2 t^{\min}_j$ holds. 
			Regarding $s \neq i$, we distinguish three cases:

			\begin{enumerate}[label=(\arabic*):]
				\item It holds $t_{sj}\leq 2t^{\min}_j,$ that is, $s\in N_j$.
				
					Then, taking \emph{all sums and products over tasks} in $N_j$
					\begin{align*}
						x_{ij}t_{ij}
						&=\frac{1/(t_{ij})^2}{\sum_{k\in N_j}1/(t_{kj})^2}\cdot t_{ij}\\
						&=\frac{\Pi_{ k\neq i}t_{kj}^2}{\sum_{l\in N_j}\Pi_{k\neq l}t_{kj}^2}\cdot t_{ij}\\
						&=\frac{t_{ij}t_{sj}\Pi_{ k\neq i,r}t_{kj}^2}{\sum_{l\in N_j}\Pi_{k\neq l}t_{kj}^2}\cdot t_{sj}\\
						&=\frac{t_{ij}t_{sj}}{t_{ij}^2+t_{sj}^2+\sum_{l\neq i,s}t_{ij}^2t_{sj}^2/t_{lj}^2}\cdot t_{sj}\\
						&\leq \frac{t_{ij}t_{sj}}{t_{ij}^2+t_{sj}^2}\cdot t_{sj}\\
						&\leq \frac{1}{2}\cdot t_{sj}.
					\end{align*}
					
					\item $t_{sj}\geq 2t^{\min}_j,$ and $t_{ij}>t_j^{\min}.$ 
					
					Then there exists a machine $k\in N_j$ with $k\neq i,s,$ and $t_{kj}=t_j^{\min}\leq t_{sj}/2.$ 
					By replacing $t_{sj}$ by $t_{kj}$ in the argument above we obtain
					$x_{ij}t_{ij} \leq t_{kj} / 2\leq t_{sj}/4$.
					
					\item $t_{sj}\geq 2t^{\min}_j,$ and $t_{ij}=t_j^{\min}.$
					
					Then we obtain directly $x_{ij}t_{ij}\leq t_{ij}=t_j^{\min}\leq t_{sj}/2.$
			\end{enumerate}
		\end{proof} 
	Finally, by summing over all tasks in $J_s$ we obtain $\sum_{j\in J_s}x_{ij}t_{ij}\leq \frac{1}{2}\sum_{j\in J_s}t_{sj}\leq \frac{1}{2} \Opt{t}.$
\end{proof}

\end{document}